\documentclass[letterpaper, 10 pt, journal, twoside]{IEEEtran}
\usepackage[utf8]{inputenc}
\usepackage{booktabs}
\usepackage{multirow}
\usepackage{adjustbox}   % for resizing tables if desired

\usepackage[british]{babel}
\usepackage{hhline}
\usepackage{multirow}

\usepackage{amsmath,amsfonts}
\usepackage[ruled,vlined, noend]{algorithm2e}

\usepackage{array}
\usepackage[caption=false,font=normalsize,labelfont=sf,textfont=sf]{subfig}
\usepackage{textcomp}
\usepackage{stfloats}
\usepackage{url}
\usepackage{verbatim}
\usepackage{tabularx}
\usepackage{graphicx}
\usepackage{cite}
\usepackage{enumerate}
\usepackage{hyperref}
\usepackage{tikz}
\usepackage{xcolor}
\usepackage{adjustbox}
\usetikzlibrary{calc}
\newcommand{\tikzmark}[1]{\tikz[overlay,remember picture,baseline] \node [anchor=base] (#1) {};}%

\usepackage{dsfont} %% Enables \mat
\usepackage{amsmath}
\usepackage{amssymb}
\usepackage{amsfonts,amsthm}       % blackboard math symbols
\usepackage{amsmath,amsbsy,amssymb,mathtools}
\usepackage{nicefrac}       % compact symbols for 1/2, etc.
\usepackage{microtype}      % microtypography
\usepackage{tikz}

\newtheorem{remark}{Remark}

\usepackage{url}
\usepackage{color}
\usepackage{subcaption}
\usepackage[utf8]{inputenc}
\usepackage[T1]{fontenc}%\usepackage{clomment}
\usepackage{balance}
\usepackage{algorithmic}
\usepackage[capitalize]{cleveref}
\crefformat{equation}{(#2#1#3)}

\newtheorem{theorem}{Theorem}

\usepackage{comment}
\usepackage{mathrsfs}
\usepackage[textfont=md,font=footnotesize]{caption}
\newcommand{\doi}[1]{\href{http://dx.doi.org/#1}{\normalsize{\textsc{doi:}}~\nolinkurl{#1}}}
\newcommand{\arxiv}[1]{\href{http://arxiv.org/abs/#1}{\normalsize{\textsc{arxiv:}}~\nolinkurl{#1}}}

\renewcommand{\epsilon}{\varepsilon}

\let\originalleft\left
\let\originalright\right
\renewcommand{\left}{\mathopen{}\mathclose\bgroup\originalleft}
\renewcommand{\right}{\aftergroup\egroup\originalright}

\def\clap#1{\hbox to 0pt{\hss#1\hss}}

\definecolor{darkmidnightblue}{rgb}{0.0, 0.2, 0.4}
\DeclareMathAlphabet{\mathpzc}{OT1}{pzc}{m}{it}

\newcommand{\jlnote}[1]%
    {\textcolor{orange}{\textbf{JL: #1}}}
\newcommand{\dlnote}[1]%
    {\textcolor{cyan}{\textbf{DL: #1}}}

\newcommand{\traj}{{\xi_{x_0}^{\pi,\phi}}}
\newcommand{\Lfx}{{L_{f_x}}}
\newcommand{\Lfd}{{L_{f_d}}}
\newcommand{\Lfxtau}{{L_{f_x}^\tau}}
\newcommand{\Lfxt}{{L_{f_x}^t}}
\newcommand{\Lf}{{L_f}}

\newcommand{\Aupper}{{\hat{A}}}
\newcommand{\Bupper}{{\hat{B}}}
\newcommand{\Dupper}{{\hat{D}}}
\newcommand{\Alower}{{\check{A}}}
\newcommand{\Blower}{{\check{B}}}
\newcommand{\Dlower}{{\check{D}}}
\newcommand{\Cupper}{{\hat{c}}}
\newcommand{\Clower}{{\check{c}}}

\newcommand{\nominalx}{{\bar{x}}}
\newcommand{\nominalu}{{\bar{u}}}

\newcommand{\neighborset}{{\mathcal{E}}}

\newcommand{\nominaldisturbancetrajx}{{\tilde{x}}}

\newcommand{\xerr}{{\epsilon_x}}
\newcommand{\derr}{{\epsilon_d}}

\newcommand{\safetyconstraintindex}{{\mathcal{I}}}

\newcommand{\NNcontrolpolicy}{{\pi_{\theta_u}}}
\newcommand{\NNdisturbancepolicy}{{\phi_{\theta_d}}}
\newcommand{\NNvalue}{{V_\theta}}
\newcommand{\NNQ}{{Q_{\theta_q}}}
\newcommand{\NNQparam}{{\theta_q}}
\newcommand{\NNcontrolparam}{{\theta_u}}
\newcommand{\NNdisturbanceparam}{{\theta_d}}

\newcommand{\reachavoidmeasurecombined}{{g(\xi_{x_0}^{\pi, \phi}, t)}}
\newcommand{\statetraj}{{\xi_{x_0}^{\pi,\phi}}}

\newcommand{\statetrajt}{{x_t}}
\newcommand{\statetrajtau}{{x_\tau}}
\newcommand{\xdist}{{\mathbb{P}}}

\newcommand{\controlspace}{{\mathcal{U}}}
\newcommand{\disturbancespace}{{\mathcal{D}}}
\newcommand{\controlpolicy}{{\pi}}
\newcommand{\disturbancepolicy}{{\phi}}

\newcommand{\valuefunc}{{V_\gamma(x)}}
\newcommand{\valuepure}{{V_\gamma}}

\newcommand{\controldim}{{\mathbb{R}^{m_u}}}
\newcommand{\disturbancedim}{{\mathbb{R}^{m_d}}}
\newcommand{\targetset}{{\mathcal{T}}}
\newcommand{\constraintset}{{\mathcal{C}}}

\newcommand{\surrogateconstraintset}{{\check{\mathcal{C}}}}
\newcommand{\surrogatetargetset}{{\check{\mathcal{T}}}}

\newcommand{\reward}{{r}}
\newcommand{\constraint}{{c}}

\newcommand{\raset}{{\mathcal{R}}}
\newcommand{\learnedraset}{{\hat{\mathcal{R}}}}

\newcommand{\Bellmanbackup}{{B_\gamma}}
\newcommand{\valuefuncone}{{V_\gamma^1}}
\newcommand{\valuefunctwo}{{V_\gamma^2}}

\newcommand{\certifiedset}{{\mathcal{S}}}
\newcommand{\reachavoidmeasure}{{g}}
\newcommand{\discountedreachavoidmeasure}{{g_\gamma}}
\newcommand{\classicalvalue}{{\bar{V}}}

\newcommand{\superzerolevelsetofV}{{\mathcal{V}_\gamma}}

\newcommand{\tone}{{t_1}}
\newcommand{\ttwo}{{t_2}}
\newcommand{\pione}{{\pi_1}}
\newcommand{\pitwo}{{\pi_2}}

\newcommand{\reachavoidcertificateL}{{\check{V}^L_{\gamma}}}
\newcommand{\reachavoidcertificateqp}{{\check{V}^{S}_{\gamma}}}

\newcommand{\rlowerLPt}{{\check{r}_t^L}}
\newcommand{\clowerLPt}{{\check{c}_t^L}}
\newcommand{\clowerLPtau}{{\check{c}_\tau^L}}
\newcommand{\rlowersocpt}{{\check{r}_t^S}}
\newcommand{\clowersocpt}{{\check{c}_t^S}}
\newcommand{\clowersocptau}{{\check{c}_\tau^S}}

\newcommand{\dynamicsfeasiablesetLP}{{\mathcal{X}_{t,\bar{x}_0}^L}}
\newcommand{\dynamicsfeasiablesetsocp}{{\mathcal{X}_{t,\nominalx_0}^S}}

\newcommand{\revise}[1]{\textcolor{black}{#1}}
\newcommand{\newrevise}[1]{\textcolor{black}{#1}}

\ifCLASSINFOpdf
\else
\fi
\begin{document}
%
% paper title
% Titles are generally capitalized except for words such as a, an, and, as,
% at, but, by, for, in, nor, of, on, or, the, to and up, which are usually
% not capitalized unless they are the first or last word of the title.
% Linebreaks \\ can be used within to get better formatting as desired.
% Do not put math or special symbols in the title.

\title{
Certifiable Reachability Learning Using a New Lipschitz Continuous Value Function
}%
%
% author names and IEEE memberships
% note positions of commas and nonbreaking spaces ( ~ ) LaTeX will not break
% a structure at a ~ so this keeps an author's name from being broken across
% two lines.
% use \thanks{} to gain access to the first footnote area
% a separate \thanks must be used for each paragraph as LaTeX2e's \thanks
% was not built to handle multiple paragraphs
%

% \author{Michael~Shell,~\IEEEmembership{Member,~IEEE,}
%         John~Doe,~\IEEEmembership{Fellow,~OSA,}
%         and~Jane~Doe,~\IEEEmembership{Life~Fellow,~IEEE}% <-this % stops a space
% \thanks{M. Shell was with the Department
% of Electrical and Computer Engineering, Georgia Institute of Technology, Atlanta,
% GA, 30332 USA e-mail: (see http://www.michaelshell.org/contact.html).}% <-this % stops a space
% \thanks{J. Doe and J. Doe are with Anonymous University.}% <-this % stops a space
% \thanks{Manuscript received April 19, 2005; revised August 26, 2015.}}
\author{
Jingqi Li$^{1}$, Donggun Lee$^{2}$, Jaewon Lee$^{3}$, Kris Shengjun Dong$^{1}$, Somayeh Sojoudi$^{1}$, Claire Tomlin$^{1}$
\thanks{
%Manuscript received: August 30, 2024; Revised November 27th, 2024; Accepted January 6th, 2025. This paper was recommended for publication by Editor Jens Kober upon evaluation of the associate editor and reviewers' comments. 
This work was supported by DARPA under the Assured Autonomy \revise{(grant FA8750-18-C-0101)} and ANSR programs \revise{(grant FA8750-23-C-0080)}, the NASA ULI program in Safe Aviation Autonomy \revise{(grant 62508787-176172)}, and the ONR Basic Research Challenge in Multibody Control Systems \revise{(grant N00014-18-1-2214)}. \revise{This work was also supported by %the NSF, ONR, the UC Noyce Initiative, and 
U.S. Army Research Laboratory and Research Office (grant W911NF2010219).}
}
\thanks{$^{1}$J. Li, K. Dong, S. Sojoudi and C. Tomlin are with University of California, Berkeley, CA 94704, USA. {\tt\footnotesize jingqili@berkeley.edu, krisdong@berkeley.edu, sojoudi@berkeley.edu, tomlin@berkeley.edu}.} 
\thanks{$^{2}$D. Lee is with the North Carolina State University, NC 27606, USA. {\tt\footnotesize dlee48@ncsu.edu}.} 
\thanks{$^{3}$J. Lee is with Boson AI, CA 95054, USA. {\tt\footnotesize lonj7798@gmail.com}. }
\thanks{Digital Object Identifier (DOI): see top of this page.}
}
\markboth{IEEE Robotics and Automation Letters. Preprint Version. Accepted January, 2025}
{Li \MakeLowercase{\textit{et al.}}: Certifiable Reachability Learning}

% The only time the second header will appear is for the odd numbered pages
% after the title page when using the twoside option.
% 
% *** Note that you probably will NOT want to include the author's ***
% *** name in the headers of peer review papers.                   ***
% You can use \ifCLASSOPTIONpeerreview for conditional compilation here if
% you desire.

% If you want to put a publisher's ID mark on the page you can do it like
% this:
%\IEEEpubid{0000--0000/00\$00.00~\copyright~2015 IEEE}
% Remember, if you use this you must call \IEEEpubidadjcol in the second
% column for its text to clear the IEEEpubid mark.

% use for special paper notices
%\IEEEspecialpapernotice{(Invited Paper)}

% make the title area
\maketitle

% As a general rule, do not put math, special symbols or citations
% in the abstract or keywords.
\begin{abstract}
We propose a new reachability learning framework for high-dimensional nonlinear systems, focusing on \emph{reach-avoid problems}. These problems require computing the \emph{reach-avoid set}, which ensures that all its elements can safely reach a target set \revise{despite disturbances} within pre-specified bounds. Our framework has two main parts: offline learning of a newly designed reach-avoid value function, and post-learning certification. Compared to prior work, our new value function is Lipschitz continuous and its associated Bellman operator is a contraction mapping, both of which improve the learning performance. %Compared to prior works, our new value function is Lipschitz continuous and the Bellman backup is a contraction mapping, both of which improve the learning performance. %Additionally, the control policy derived from our value function reaches the target set more rapidly than those from prior works. 
To ensure deterministic guarantees of our learned reach-avoid set, we introduce two efficient post-learning certification methods. %One method leverages the Lipschitz constants and the learned control policy to certify a subset of the ground truth reach-avoid set, while the other uses second-order cone programming. 
Both methods can be used online for real-time local certification or offline for comprehensive certification. We validate our framework in a 12-dimensional crazyflie drone racing hardware experiment and a simulated 10-dimensional highway take-over example. %, demonstrating its capability to compute reach-avoid sets for high-dimensional systems with deterministic guarantees.%However, it is challenging to guarantee the convergence of DDPG. We  %Moreover, different from prior works on ensuring the safety of a state via online safety filter, we propose two methods for certifying whether all states in a user-defined set can reach a target set safely for all possible disturbances, with deterministic guarantees. %Featured with a newly designed reachability value function, our method offers higher computational efficiency than prior works. This improved computational efficiency results from two beneficial properties of our new value function. First, our value function is Lipschitz continuous, and this helps training stability when using neural networks for approximating the value functions.
\end{abstract}

% Note that keywords are not normally used for peerreview papers.
% \begin{IEEEkeywords}
% IEEE, IEEEtran, journal, \LaTeX, paper, template.
% \end{IEEEkeywords}
\begin{IEEEkeywords}
Reachability, Machine Learning for Robot Control, Reinforcement Learning
\end{IEEEkeywords}

% For peer review papers, you can put extra information on the cover
% page as needed:
% \ifCLASSOPTIONpeerreview
% \begin{center} \bfseries EDICS Category: 3-BBND \end{center}
% \fi
%
% For peerreview papers, this IEEEtran command inserts a page break and
% creates the second title. It will be ignored for other modes.
\IEEEpeerreviewmaketitle

% \section{Introduction}

% \begin{enumerate}
%     \item Reachability analysis is crucial for ensuring safety of autonomous systems.
%     \item Classical methods based on dynamic programming suffers from curse of dimensionality.
%     \item To accommodate this challenge, recent work propose to learn neural network value function.
%     \item When the learned value function converged, we obtain the ground truth reach-avoid set. However, in practice, it's hard to guarantee the convergence of learning algorithm.
%     \item Since it's challenging to ensure the convergence of value function, our key idea is to calibrate the learned reach-avoid set from some sub-optimal value function.
%     \item To the best of the authors' knowledge, our work is the only work that could guarantee exact safety of the learned reach-avoid set, compared with the state-of-the-art reachability learning method. 
%     \item we utilize Lipschitzness for safety verification
%     \item We propose an efficient method for learning reachable set for nonlinear systems
%     \item we compare our method with deep Reach and Jaime's work.
% \end{enumerate}
\section{Introduction}
\label{sec:introduction}

% The thing we want to say in Introduction
% 1. Why reachability analisys matters
% 2. 

% Ensuring the safety and performance of autonomous systems is essential for safety-critical applications, such as autonomous driving \cite{muhammad2020deep}, surgical robots \cite{thananjeyan2020safety}, and air traffic control \cite{hwang2007protocol}. Many of those problems could be modeled as controlling the system to a target set in the state space without violating safety constraints \cite{bansal2017hamilton}. This is referred to as the \emph{reach-avoid} problem \cite{summers2011stochastic}. However, it is challenging to solve this problem under the presence of uncertainty \cite{dorato1987historical}, such as modeling errors and environmental disturbances. One way to accommodate this is to consider the reach-avoid problem under the worst-case scenario, which could be formulated as a zero-sum game between the control inputs and an adversary, accounting for the uncertainty or disturbance, with an objective to compromise the efforts of control inputs\cite{tomlin1998conflict}. 
\IEEEPARstart{E}{nsuring} the safe and reliable operation of robotic systems in %dynamic and 
uncertain environments is a critical challenge as autonomy is introduced into everyday systems. For instance, we would like humanoid robots to safely work close to humans.  As a second example, new concepts for air taxis will need real-time synthesis of safe trajectories in crowded airspace. These safety-critical applications are typically characterized by sequences of tasks, %it is essential to complete various tasks safely despite unpredictable disturbances. Additionally, 
and knowing the set of states from which the task can be safely completed despite unpredictable disturbances is important.  %because the task can fail once the robot exits that set. 
Reachability analysis addresses this challenge by determining the \emph{reach-avoid set}—a set of states that can safely reach a target set under all possible disturbances within a specified~bound, as well as the corresponding control.% to use on the boundary of the set.

Traditional~Hamilton-Jacobi reachability analysis methods \cite{tomlin1998conflict,lygeros2011differentialgameshj,fisac2015reach} leverage dynamic programming %approach 
to synthesize the optimal control and a reachability value function, whose sign indicates whether or not a state can safely reach the target set. Though theoretically sound, they suffer from the %significant practical limitations. Chief among these is the 
\emph{curse of dimensionality} \cite{bellman1957dynamic}:  as the system's dimension increases, the computational complexity grows exponentially, making these methods impractical for real-world applications without approximation or further logic to manage the problem size.

There~has~been~interest~in using machine learning techniques to estimate reachability value functions for high-dimensional systems \cite{bansal2020deepreach, fisac2019bridging, hsusafety, hsu2023sim, hsu2023safety, hsu2023isaacs, li2023learning}. However, a major drawback of \revise{existing~reachability~learning~methods} is the lack of deterministic safety guarantees. 
%Without such guarantees, it is difficult to ensure the reliability and safety of the robotic systems in critical scenarios.
Recent work \cite{lin2023generating,lin2024verification} provides probabilistic safety guarantees for the learned reach-avoid sets. %, which require a large number of samples to achieve a high level of confidence. While these methods can be useful, they are often impractical for safety-critical applications where deterministic guarantees are necessary. 
Additionally, safety filter approaches \cite{hsu2023isaacs,li2023learning, nguyen2024gameplay} have been proposed, which offer point-wise guarantees by ensuring safety for individual states. %However, these methods fall short as they do not guarantee safety for the entire learned reach-avoid set, necessitating the enumeration of an infinite number of states in a continuous set, which is infeasible.

\revise{In this work, %we ask if it is possible to learn a set of states that have deterministic reach-avoid guarantees. 
we propose the first method for learning reach-avoid sets for high-dimensional nonlinear systems with deterministic assurances. Our method involves learning a new reach-avoid value function (Sections~\ref{sec:methods} and \ref{sec:learning}), and then conducting set-based certification to ensure that all states in the certified set safely reach the target set despite disturbances (Section~\ref{sec:certification}).} Specifically:%, our contributions include: 

\begin{figure}[t]
    \centering
    \includegraphics[width=0.49\textwidth]{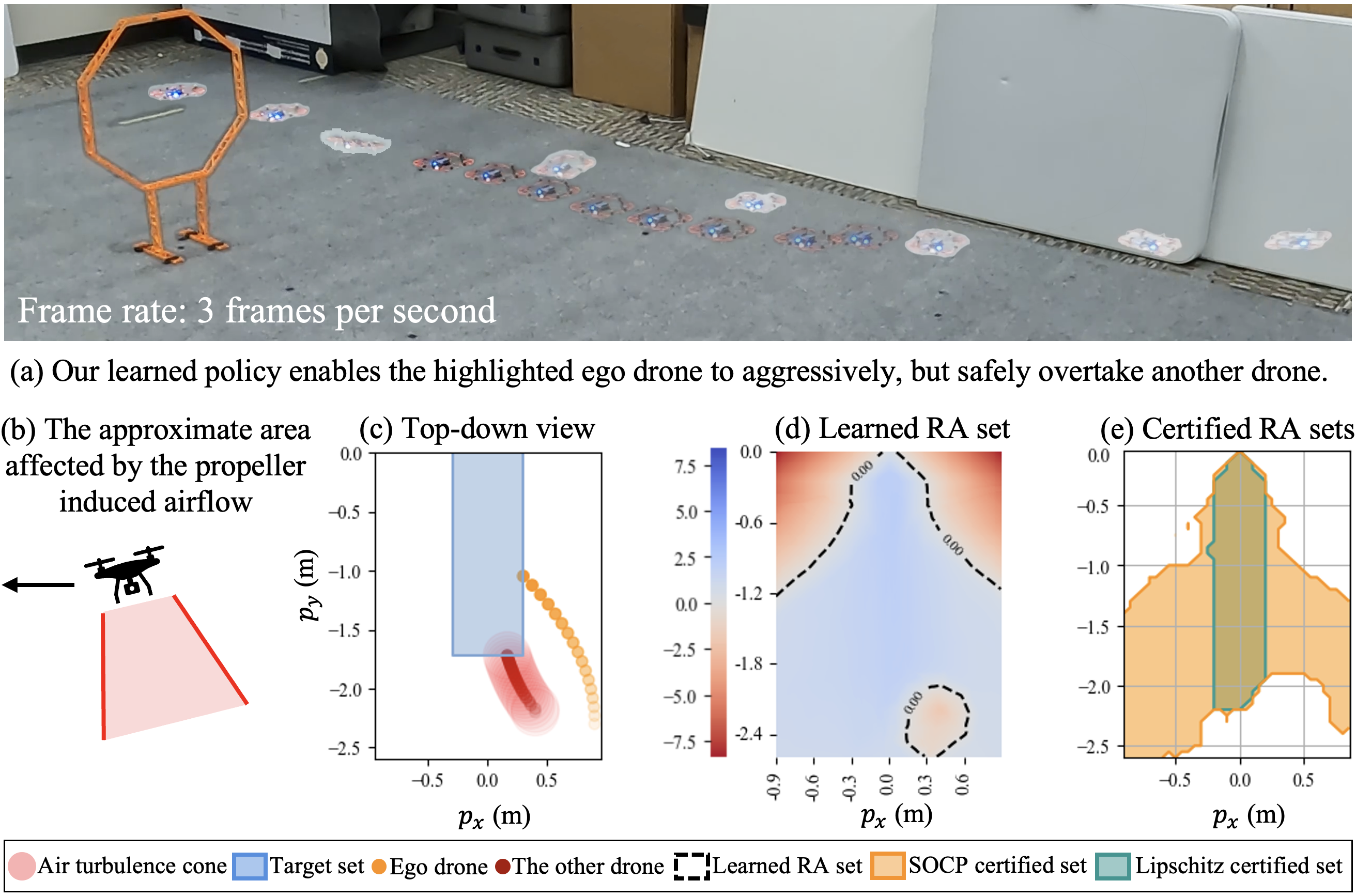}\vspace{-0.4em}
    \caption{Applying our reachability analysis framework to drone racing. %The ego drone, highlighted in the top subfigure, aims to safely overtake another drone before flying through a vertical gate located at the origin, under the worst-case disturbance in another drone’s acceleration. 
    In (a), hardware experiments demonstrate that our learned control policy enables an ego drone to safely overtake another drone, despite unpredictable disturbances in the other drone's acceleration. 
    In (b), we illustrate the concept of the propeller induced airflow \cite{flem2024experimental}, which can affect other drones' flight. In (c), we apply our learned control policy in a simulation with randomly sampled disturbances.
    In (d), we project the learned reach-avoid value function onto the $(x,y)$ position of the ego drone. The super-zero level set, outlined by dashed curves, indicates our learned reach-avoid (RA) set. In (e), we plot the certified RA sets using Lipschitz and second-order cone programming certification.}\vspace{-1.3em}
    \label{fig:front}
\end{figure}

% [Our approach]

%To address these limitations, 
% To this 

1) \revise{We propose a new reach-avoid value function that is provably Lipschitz continuous. \revise{Though it is not based on a Lagrange-type objective function} (cumulative rewards over time) as
in classical reinforcement learning (RL) \cite{bertsekas2012dynamic}, we prove its
Bellman equation is still a contraction mapping, removing the need for contractive Bellman equation approximation
commonly used in prior reachability works \cite{hsu2023isaacs,li2023learning}.
Moreover, the control policy derived from
our value function tends to reach the target set rapidly.
We apply deep RL to learn this new value function.}

2) We develop two reach-avoid set certification methods. The first uses the Lipschitz constant of the dynamics to certify the safety of a subset of the learned reach-avoid set, ensuring all its elements can safely reach the target set \revise{under} disturbances. The second employs second-order cone programming to do the same. Both methods offer deterministic assurances. %even when there are disturbances, defined within a prespecified bound. 
They can be applied online to verify if a neighboring set around the current state can safely reach the target set, \newrevise{%despite disturbances within a pre-specified set, 
or offline for verifying 
the same for %whether all states in 
a larger user-defined set.} %can safely reach the target set. }%comprehensive~certification.

3) We show the computational benefits of our new value function and the assurance of our (real-time) certification methods through simulations and hardware experiments. We empirically justify that %our learned policy compares favorably with that of prior works, and that 
    the Lipschitz continuity of our new value function can accelerate value function learning. \vspace{-1.em}%Additionally, our online certification methods can be computed in real-time. %Additionally, any state in a certified reach-avoid set is guaranteed to reach the target safely under all possible disturbances.
\section{Related works}%\vspace{-0.3em}
% \subsection{Reachability learning methods}

\noindent\textbf{Hamilton-Jacobi reachability learning.} DeepReach \cite{bansal2020deepreach} is a pioneering work on learning finite-horizon reachability value functions. %However, for problems with high-dimensional dynamics and complex constraints, determining an appropriate finite horizon for reachability analysis can be challenging. %finite-horizon reachability analysis could be conservativeoverlooking the possibility that a state could reach the target set in a time horizon longer than the pre-specified one makes.}
%, where we compute a set of initial states from which we can reach the target set in a pre-specified length of time. However, when the knowledge of such an upper bound on the target set reaching time is unavailable, the computed reach-avoid set could be conservative because it overlooks the possibility that a state could reach the target set in a longer time. 
In other studies, such as \cite{fisac2019bridging, hsusafety, hsu2023safety, hsu2023sim, hsu2023isaacs, nguyen2024gameplay, li2023learning}, the assumption of a known horizon is relaxed and infinite horizon reachability learning problems are considered. In this work, we also consider the infinite horizon case. We introduce a new value function which is provably Lipschitz continuous and whose Bellman operator is a contraction mapping, offering computational efficiency when compared with prior work.%the assumption of a known horizon is relaxed and infinite horizon reachability learning problems are considered. In this work, we also consider the infinite horizon case.  Our work introduces a new value function which is provably Lipschitz continuous and whose ground truth Bellman equation is a contraction mapping (Section \ref{sec:methods}). This offers higher computational efficiency compared to prior works. %In contrast, we relax the assumption of knowing the horizon and consider infinite-horizon reachability problems. Other studies \cite{fisac2019bridging, hsusafety, hsu2023safety, hsu2023sim, hsu2023isaacs, nguyen2024gameplay, li2023learning} also consider infinite horizon reachability learning problems. %, using deep reinforcement learning (RL) to learn reachability value functions. 

\noindent\textbf{Verification of learning-based control}. %Recent works \cite{lin2023generating,lin2023verification,singh2024imposing} have extended DeepReach to provide probabilistic safety guarantees, using a large number of samples and offline computation. %, which require many random samples to achieve high confidence but do not address systems with disturbances.
%However, both of our two certification methods provide deterministic guarantees and they can be computed in real-time. % that all elements of the certified reach-avoid set reaches the target set safely under the worst-case disturbances. Moreover, our certificatino can be done in real-time. %We use the learned policy to construct theoretical lower bounds of the reachability value function, eliminating the need for numerous samples.
%Recent work has \cite{lin2023generating, lin2023verification, singh2024imposing} provided probabilistic safety guarantees for DeepReach using a large number of samples and offline computation. In contrast, our certification methods offer deterministic guarantees and can be computed in real-time. 
Recent work \cite{lin2023generating,lin2024verification,singh2024imposing} has provided probabilistic safety guarantees for DeepReach.  In this paper, we introduce methods that provide deterministic reach-avoid guarantees and we show how they could be used locally in real time. Other studies \cite{hsu2023isaacs, nguyen2024gameplay, li2023learning} provide point-wise safety filters. However, we propose set-based reach-avoid certification methods to verify if all states in a set can safely reach the target set \revise{under} potential disturbances. Our certification methods also differ from existing set-based approaches for verifying neural network-controlled systems, including those that verify regions of attraction \cite{nguyen2021stability, korda2022stability, schwan2023stability}, \revise{forward reachability sets \cite{coogan2020mixed,lew2021sampling ,%xiang2018reachability 
everett2021efficient, hu2020reach, kwonconformalized}, and safe sets using barrier certificates \cite{prajna2004safety, mazouz2022safety, xiao2023barriernet}. To the best of the authors’ knowledge,  our work is the first
to certify if a set of initial states is
within the ground truth reach-avoid set. }%To the best of the authors' knowledge, these methods estimate the forward reachable set of a given set of initial states, but we certify if a set of initial states is within the ground truth reach-avoid set, for the first time. }%In general, these methods do not certify if a set of initial states is within the ground truth reach-avoid set, as ours~do.%

\noindent\textbf{Constrained optimal control}. 
Control barrier functions (CBFs) offer safety guarantees \cite{ames2016control, taylor2020learning, qin2021learning}, but they tend to be more conservative than model predictive control (MPC) \cite{richalet1976algorithmic, lopez2019tubempc}, which optimally balances task performance and safety. However, constructing or learning CBFs and solving MPC can be difficult for nonlinear systems with complex constraints. By leveraging deep neural networks, constrained reinforcement learning (CRL) \cite{achiam2017constrained, chow2019lyapunov, yang2021wcsac} learns control policies that maximize task rewards while adhering to complex constraints. %\newrevise{A positive CBF value indicates that the state remains safe in the future, but CBFs are typically not designed to guarantee that the state will be driven to a specific target set.} %provides safety information for a state; it does not indicate whether a trajectory from that state can reach the target set. 
\newrevise{CBFs, along with the typical value functions defined in MPC and CRL, do not provide the information about whether a state can safely reach the target set.} In contrast, our new value function not only provides the optimal control for the worst-case disturbance, but also indicates, based on its sign, whether or not a state \revise{can safely reach the target set}.%\vspace{-1em} %Additionally, we will empirically show that policies learned with our new reach-avoid value function achieve higher success rates in take-over tasks compared to CRL baselines.% This suggests that our policy is better suited for challenging reach-avoid tasks than CRL baselines. %This results enlightens an interesting future work of further characterization of different value function and their optimal policies. %Our work also enlightens a future direction that further characterize the different policies derived from different value functions. 

%In addition, there is another line of works using control barrier function for regulating an optimal control for ensuring safety. However, CBF could be hard to constructed for complex constraints. Moreover, for the learning-based CBF work, our method learns reach-avoid value functions that not only cares about safety but also optimizes the target set reaching. 

% Pointwise safety: Roll out trajectory. safety filter. 

% Set safety:

% Trustworthy certifcation. Most of the existing works on neural network verification certify whether the forward reachable set intersects with the unsafe set. There are also some works on verifying the stability of the closed-loop control systems involving neural network controllers. To the best of the author's knowledge, our work is the first on verifying whether a continuous set of states can be driven to the target set safely under a neural network controller. 

% Untrustworthy certification. Learning certificate like NN CBF. NN LYF. We provide deterministic guarantee. 

\section{Problem Formulation}%\vspace{-0.3em}
\label{sec:problem_formulation}
\revise{We consider uncertain nonlinear dynamics described by }\vspace{-0.3em}
\begin{equation}
    x_{t+1} = f(x_t, u_t, d_t),\vspace{-0.3em}
    \label{eq:dynamics}
\end{equation}
where $x_t \in \mathbb{R}^n$ is the state, $u_t\in \controlspace\subseteq \controldim$ is the control, %(typically from a control policy $\controlpolicy:\mathbb{R}^n \to \controlspace$), 
and $d_t\in\disturbancespace\subseteq\disturbancedim$ represents the disturbance, such as model mismatch or uncertain actions of other agents. %other uncertainties. 
\revise{We assume that both $\controlspace$ and $\disturbancespace$ are compact and connected sets. The disturbance bound could be estimated from prior data or a physical model}. %We also assume the dynamics $f$ is $\Lf$ Lipschitz continuous.
%We also denote by $\controlpolicy(\cdot):\mathbb{R}^n\to \controlspace$ a control policy and $\disturbancepolicy(\cdot, \cdot): \mathbb{R}^n\times \controlspace\to \disturbancespace $ a disturbance policy, conditioned on the state $x$ and the control $u$. 
%We define a state trajectory originating from $x_0$, controlled by $\pi$ and $\phi$, as $\traj: = \{x_t\}_{t=0}^\infty$, where $x_{t+1}= f(x_t, \controlpolicy(x_t), \disturbancepolicy(x_t, \controlpolicy(x_t)))$, $\forall t$.
%We denote by $\controlpolicy(\cdot):\mathbb{R}^n\to \controlspace$ a control policy and $\disturbancepolicy(\cdot, \cdot): \mathbb{R}^n\times \controlspace\to \disturbancespace $ a disturbance policy, conditioned on a state $x$ and a control $u$. 
We define a state trajectory originating from \revise{an initial condition} $x_0$ under a control policy $\revise{\pi}:\mathbb{R}^n \to \controlspace$ and a disturbance policy $\revise{\phi}: \mathbb{R}^n\times \controlspace\to \disturbancespace$ as $\traj: = \{x_t\}_{t=0}^\infty$, where $x_{t+1}= f(x_t, \controlpolicy(x_t), \disturbancepolicy(x_t, \controlpolicy(x_t)))$,\revise{~$\forall t \in \{0,1,2,\cdots\}$}. %where the function $f: \mathbb{R}^n\times \controlspace \times \disturbancespace \to \mathbb{R}^n$ is assumed to be Lipschitz continuous. 

\revise{Let \(\targetset \subseteq \mathbb{R}^n\) be an open set, representing a \emph{target set}.} %Without loss of generality, we can associate $\targetset$ with 
We assume that there exists a Lipschitz continuous, bounded reward function $\revise{r}:\mathbb{R}^n\to \mathbb{R}$ indicating if a state $x$ is in the target set,\vspace{-0.3em} 
\begin{equation}
r(x)>0\Longleftrightarrow x\in \targetset. \vspace{-0.3em}
% \textrm{Int}(\targetset). 
\end{equation}
We consider a finite number of Lipschitz continuous, bounded constraint functions $c_i(x)>0,\forall i \in \safetyconstraintindex$, where $\safetyconstraintindex$ is the set of indexes of constraints. \revise{Throughout this paper, we define Lipschitz continuity with respect to the \revise{$\ell_2$} norm.} We can simplify the representation of constraints by considering their pointwise minimum $c(x):=\min_{i\in\safetyconstraintindex}c_i(x)$. We define the constraint set as $\constraintset := \{x\in\mathbb{R}^n:c(x)>0 \}$, and we have\vspace{-0.3em}
\begin{equation}
    c(x)>0 \Longleftrightarrow x\in \constraintset.\vspace{-0.3em}
\end{equation}
We look for states that can be controlled to the target set safely under \emph{the worst-case disturbance}, with dynamics given by \eqref{eq:dynamics}. %, i.e., under all possible disturbance polices $\disturbancepolicy:\mathbb{R}^n\times \controlspace\to \disturbancespace$. 
We refer to this set as the \textbf{\emph{reach-avoid (RA) set}} \cite{tomlin1999computing,lygeros2011differentialgameshj}: 
%Formally, we define the reach-avoid set as %for the dynamic systems whose trajectory is driven by a control and arbitrary disturbance policy: 
% $x_{t+1}=f(x_t,\pi(x_t),\phi(x_t,\pi(x_t)))$.  }
\begin{equation}\label{eq:raset}
    \raset:=\left\{ \begin{aligned}x_0 :& \exists \pi\ \textrm{such that } \forall \phi,\  \exists T<\infty, \\ 
    %&x_{t+1}=f(x_t,\pi(x_t),\phi(x_t,\pi(x_t))),\forall t\\ 
    &\big(r(x_T)>0 \ \wedge\  \forall t\in[0,T], c(x_t)>0\big) \end{aligned}\right\},
    % \raset := \Big\{&x_0: \forall \nonanticipativestrategy ,\ \exists \controlseq \textrm{ and }T<\infty \textrm{ such that }\ \disturbanceseq = \nonanticipativestrategy(\controlseq),    \\&  \reward(
    % x_T%\statetraj(T)
    % ) >0  \  \wedge \    \big(\forall t\in [0,T],\ \constraint(
    % x_t%\statetraj(t)
    % )>0\big) \Big\},
\end{equation}
%\textcolor{red}{which includes all the states that some policy $\pi$ can drive to the target set safely in finite time despite any disturbance policies $\phi$. }
% \begin{equation}
%     \{x: \exists \pi:x\to u, \forall \phi : (x,u)\to d , \exists T<0, \}
% \end{equation}
which includes all the states that can reach the target set safely in finite time \revise{despite disturbances} within the set $\disturbancespace$. %In $\eqref{eq:raset}$, the state $x_t$ evolves according to the closed-loop dynamics $x_{t+1} = f(x_t, u_t,d_t)$, where $u_t = \controlpolicy(x_t)$ and $d_t = \disturbancepolicy (x_t, u_t)$, for all $t$. %, even under the case where we have the worst-case disturbance response $d_t$ to the control inputs $\{u_\tau\}_{\tau=0}^t$ at all the times $t\in\mathbb{Z}_+$.%, there exists a control sequence $\controlseq$ driving the state to the target set safely. 

% just use u and d
% \textcolor{red}{
% \begin{equation*}
%     \begin{aligned}
%     \raset := \{&x\in\mathbb{R}^n: \exists u_0 s.t. [ \forall d_0, \exists u_1 s.t. [...., \exists T>0 such that ] ]     \\&  \statetraj(T) \in\targetset \  \wedge \    \forall t\in [0,T],\ \constraint(\statetraj(t))>0 \}.
%     \end{aligned}
% \end{equation*}
% }
% non-anticipative strategy

% In this paper, we compute the RA set for high-dimensional nonlinear systems with complex constraints. %In particular, we evaluate a set of states that reach the target safely under the worst-case disturbance, with deterministic guarantees. 

\noindent\textbf{Running example, safe take-over in drone racing:}  We model the drone take-over example in Figure~\ref{fig:front} as an RA problem, % with complex constraints. 
%We will present a more complicated nonlinear RA problem in Section~\ref{sec:experiment}. 
where two crazyflie drones \cite{crazyflie2017} compete to fly through an orange gate. % (30 cm by 30 cm). 
The first drone (ego agent) starts behind and aims to overtake the other drone. % before reaching the gate. 
The second drone flies directly to the gate using an LQR controller, but its acceleration is uncertain to the first drone. We model the uncertain part of the other drone's acceleration by a disturbance $\|d_t \|_2 \le\epsilon_d := 0.1\textrm{ m/s}^2$. We compute the RA set to ensure the ego drone can safely overtake the other despite this disturbance. %in the other drone’s acceleration. %despite any disturbance $\|d_t\|_2\le \epsilon_d$. 
%We use a hierarchical control framework. 

We consider a 12-dimensional dynamics \cite{pichierri2023crazychoir}, where the $i$-th drone's state is $x_t^i=[p_{x,t}^i, v_{x,t}^i, p_{y,t}^i, v_{y,t}^i, p_{z,t}^i, v_{z,t}^i]$. In each of the $(x,y,z)$ axes, the $i$-th drone is modeled by double integrator dynamics, and the control is its acceleration $u_t^i=[a_{x,t}^i,a_{y,t}^i, a_{z,t}^i]$, with $\|u_t^i\|_\infty\le \epsilon_u:=1\textrm{m/s}^2$. %The low-level control is managed by well-tuned PID motor control, robustly tracking the high-level system's state~$x_t$. %, robustly tracking the high-level system's trajectory. 
We model the center of the gate as the origin. The radius of the orange gate is 0.3 meters, and the radius of the crazyflie drone is 0.05 meters. 
We consider a target set for the ego drone:%, as plotted in~Figure~\ref{fig:front}: % in front the other drone and within $40$ cm of a center 
% line perpendicular to the gate, 
\begin{equation}\label{eq:drone target}
    \targetset= \left\{\begin{aligned} x:\  & p_{y}^1 - p_{y}^2>0,\hspace{0.5cm} v_{y}^1 - v_{y}^2>0, \\ & |p_{x}^1|< 0.3,\hspace{0.75cm} |p_{z}^1|<0.3 \end{aligned} \right\}.
\end{equation}
To ensure the ego drone flies through the gate, we constrain:
\begin{equation}\label{eq:constraint pass}
\begin{aligned}
    \pm p_{x,t}^1 - p_{y,t}^1  > -0.05,\hspace{0.5cm}  \pm p_{z,t}^1 - p_{y,t}^1  > -0.05.
    % \| p_{x,t}^1-p_{z,t}^1 \| \le 0.05, \|p_{x,t}^1 - p_{y,t}^1\| \le 0.05. 
    %& p_{x,t}^1-p_{y,t}^1 + 0.05>0,\ 
    %&-p_{x,t}^1 - p_{y,t}^1 + 0.05>0\\
    %p_{x,t}^1 - p_{y,t}^1 + 0.05>0,\ &-p_{x,t}^1 - p_{y,t}^1 - 0.05 > 0 
\end{aligned}
\end{equation}
To ensure safe flight, the ego drone should avoid the area affected by the airflow from the other drone, as depicted in Figure~\ref{fig:front}, using the constraint: \vspace{-0.3em}
%The airflow induced by the propellers can cause other drones to crash. We consider the following safety cone constraints: 
% One drone could crash when it flies into another drones' air turbulence cone. 
\begin{equation}
    \begin{aligned}
         \left\| \begin{bmatrix}
            p_{x,t}^1 - p_{x,t}^2\\ p_{y,t}^1 - p_{y,t}^2
        \end{bmatrix}\right\|_2^2 >\Big(1+ \max(p_{z,t}^2- p_{z,t}^1,0) \Big) \times 0.2 ,
    \end{aligned}\label{eq:safety cone}\vspace{-0.3em}
\end{equation}
where the required separation distance between the ego drone and the other drone increases as their height difference grows.
% and we consider this constraint when the height of the ego drone is less than the other drone, i.e., $p_{z,t}^1 \le p_{z,t}^2$. %, $\forall i,j\in\{1,2\}$. 

%Due to the high dimensionality and nonconvex constraints of this problem, 
%Due to the curse of dimensionality, we cannot directly apply traditional Hamilton Jacobi analysis to this problem. 
%Numerically computing the RA value function directly for systems with state dimension more than 5-6 remains a challenging problem \cite{bansal2017hamilton}. 
Numerically computing the RA set directly for this problem is computationally infeasible \cite{bansal2017hamilton}. We introduce our new reachability learning method in the following sections.

\section{A New Reach-Avoid Value Function}\label{sec:methods}
%Our method is centered around a new RA value function, where the sign indicates whether a state is in the RA set. %To be more specific, we will show that the \emph{super-zero level set} of the designed value function recovers the RA set in Problem~\ref{Reach_avoid_problem}. 
%Different from the RA value functions formulated in prior works \cite{fisac2019bridging, hsusafety,hsu2023safety,hsu2023sim,hsu2023isaacs,nguyen2024gameplay,li2023learning}, 
%Different from prior works \cite{fisac2019bridging, hsusafety,hsu2023safety,hsu2023sim,hsu2023isaacs,nguyen2024gameplay,li2023learning}, we include time discount factor $\gamma\in[0,1)$ in the value function, which yields a contractive Bellman equation. By choosing appropriate $\gamma$, we guarantee the Lipschitz continuity of our new value function and the corresponding optimal control will reach the target set fast. 
In this section, we propose a new RA value function for evaluating if a state belongs to the RA set. Unlike prior works \cite{fisac2019bridging, hsusafety, hsu2023safety, hsu2023sim, hsu2023isaacs, nguyen2024gameplay, li2023learning}, our value function incorporates a time-discount factor. This results in a Lipschitz-continuous value function, which appears to accelerate the learning process, and establishes a contractive Bellman equation, eliminating the need for the contractive Bellman equation approximation commonly used in prior works \cite{fisac2019bridging, hsusafety, hsu2023safety, hsu2023sim, hsu2023isaacs, nguyen2024gameplay, li2023learning}. %This addition results in a Lipschitz-continuous value function, which improves learning stability, and establishes a contractive Bellman equation, which removes the need for contractive Bellman backup approximation commonly used in prior works. %—features absent in those previous works. 
Furthermore, we show that the \revise{control policy} derived from this new value function tends to reach the target set rapidly.%, attributed to the inclusion of the time-discount factor.%because they do not incorporate time-discount factors.
%Also, our formulation enables faster completion of the RA task than those previous works.

%our new value function directly includes the time-discount factor and this offers three advantages: 1) For all feasible discount factors, its super-zero level set recovers the exact RA set, and the Bellman equation is a contraction mapping, i.e., there is a unique value function satisfying the Bellman equation; 2) The value function can be constructed to be Lipschitz continuous. %, which is a favorable property for HJ analysis \cite{basco2019lipschitz}.
% 3) A small time-discount factor renders the optimal state trajectory to reach the target set fastly.  
% Our formulation enables to achieve computational efficiency, Lipschitz-continuous value function, in contrast, prior works \cite{fisac2019bridging, hsusafety,hsu2023safety,hsu2023sim,hsu2023isaacs,nguyen2024gameplay,li2023learning} result in 

% Before we introduce our new value function, we build up the intuition by first reviewing the RA value functions in prior works. 

%We will first build up the intuition behind our new value function design and then explain why it is beneficial to have time discount factor in value function. 

We begin the construction of our new value function by first introducing the concept of \emph{RA measure}, which assesses whether a trajectory can reach the target set safely. Let $\statetraj$ be a trajectory that enters the target set safely at a stage $t$. We have $\reward(\statetrajt)>0$ and $\constraint(\statetrajtau)>0$ for all $\tau \in\{0,1,\dots ,t\}$. %Let $\gamma\in[0,1]$ be the time-discount factor, which discounts the impact of future reward and constraints. 
% We observe that such a sequence of control inputs $\mathbf{u} = \{u_t\}_{t=0}^\infty$ should satisfy
In other words, the \emph{RA measure} $\reachavoidmeasurecombined$, defined as
\begin{equation}
\begin{aligned}
     \reachavoidmeasurecombined
     :=\min \Big\{  r\big(\statetrajt\big), \min_{\tau=0,\dots, t}  c\big(\statetrajtau\big) \Big\},
     %&x_{\tau+1} = f(x_\tau, \controlpolicy(x_\tau),\disturbancepolicy(x_\tau,u_\tau)),  \forall \tau\le t,
\end{aligned}
\end{equation}
is positive, $\reachavoidmeasurecombined>0$, if and only if there exists a trajectory from $x_0$ reaching the target set safely.% at a finite time $t$.
% \begin{equation}
%     \reachavoidmeasure(\statetraj,t)>0.
% \end{equation}

An \emph{RA value function} $\classicalvalue(x)$ has been proposed in prior works \cite{fisac2019bridging, hsusafety,hsu2023safety,hsu2023sim,hsu2023isaacs,nguyen2024gameplay,li2023learning}, and it evaluates the maximum RA measure %achieved by an optimal control sequence $\controlseq$ 
under the worst-case disturbance: %$\disturbanceseq$: %considers an infinite-horizon zero-game between the control $u_t$ and the disturbance $d_t$, where the control wants to maximize the optimal RA measure $\sup_{t=0,\dots} \reachavoidmeasure(\statetraj,t)$, whereas the disturbance aims to minimize it:
\begin{equation}\label{eq:define classical RA value}
    \classicalvalue(x):= %\prod_{s=0}^\infty [\max_{u_s}\min_{d_s}] 
    \max_{\controlpolicy}\min_{\disturbancepolicy}\sup_{t=0,\dots}\reachavoidmeasurecombined ,
    %\reachavoidmeasure(\statetraj,t).
\end{equation}
where $\classicalvalue(x)>0$ if and only if $x\in \raset$. 
%To verify the existence of a sequence of control inputs driving the state trajectory to the target set safely from an initial state $x_0$, it suffices to check the sign of the following term:
% \begin{equation}\label{eq:inf_time_check_finite_time}
%     \prod_{s=0}^\infty[\max_{u_s}\min_{d_s}] \sup_{t=0,\dots} \reachavoidmeasure(\statetraj,t)
%     % \begin{aligned}
%     %     \prod_{s=0}^\infty[\max_{u_s}\min_{d_s}] \sup_{t=0,\dots} \min\big\{ \gamma^t r(\xi_{x}^{\mathbf{u,d}}(t)), \min_{\tau = 0,\dots,t}\gamma^\tau c(\xi_{x}^{\mathbf{u,d}}(\tau)) \big\}.
%     % \end{aligned}
% \end{equation}
% By using the notion of the non-anticipative strategy $\phi$, we can simplify \eqref{eq:inf_time_check_finite_time} to the term
% \begin{equation*}
% \begin{aligned}
%     \inf_{\phi}\max_{\mathbf{u}}&\sup_{t=0,\dots, } \min \big\{ \gamma^t r(\xi_{x}^{\mathbf{u},\phi(\mathbf{u})}(t)), \min_{\tau=0,\dots, t} \gamma^{\tau} c(\xi_{x}^{\mathbf{u},\phi(\mathbf{u})}(\tau)) \big\}.
% \end{aligned}
% \end{equation*}
%When $\classicalvalue(x)>0$, it means the state $x$ can be driven to the. However, 
%The super-zero level of $\classicalvalue(x)$ recovers the RA set. 
%A state $x$ is in the RA set if $\classicalvalue(x)>0$. 
%A state $x$ is in the RA set whenever $\classicalvalue(x)>0$. 
We compute $\classicalvalue(x)$ by solving its Bellman equation. However, $\classicalvalue(x)$ has a non-contractive Bellman equation, whose solution may not recover $\raset$, as shown in Figure~\ref{fig:counter example continuity}. To address this, 
prior works \cite{fisac2019bridging, hsusafety, hsu2023safety, hsu2023sim, hsu2023isaacs, nguyen2024gameplay, li2023learning} include a time-discount factor $\gamma$ to create a contractive Bellman equation approximation. For each $\gamma\in(0,1)$, there is a unique solution to the approximated Bellman equation, which converges to $\classicalvalue(x)$ as $\gamma$ is gradually annealed to 1.  %The RA set is recovered by gradually annealing $\gamma\to 1$ as the solution of the approximated Bellman equation converges to $\classicalvalue(x)$. 
% However, one issue of this value function is that its Bellman equation is non-contractive \cite{fisac2019bridging, hsusafety,hsu2023safety,hsu2023sim,hsu2023isaacs,nguyen2024gameplay,li2023learning}. The non-contractiveness of the Bellman equation makes the Bellman equation ill-posed, meaning that there can be an infinite number of value functions satisfying the Bellman equation but their super-zero level sets are not necessarily to recover the RA set, as shown in the Appendix. Though an approximate but contractive Bellman equation has been proposed in \cite{fisac2019bridging} by introducing a time-discount factor $\gamma\in [0,1)$, the unique value function solution to that approximate Bellman equation yields only a conservative subset of the RA set $\raset$ for any time discount factor $\gamma<1$. It is suggested in \cite{fisac2019bridging, hsusafety,hsu2023safety,hsu2023sim,hsu2023isaacs,nguyen2024gameplay,li2023learning} that annealing the $\gamma\to 1$ will recover the exact RA set. 

Inspired by previous studies, we enhance computational efficiency by designing a new \emph{time-discounted} RA value function. Ours incorporates a time-discount factor into the value function formulation, resulting in a contractive Bellman equation without the need for any approximation. This improvement eliminates the requirement of the $\gamma$-annealing process commonly used in prior works \cite{fisac2019bridging, hsusafety, hsu2023safety, hsu2023sim, hsu2023isaacs, nguyen2024gameplay, li2023learning}, \revise{where $N$ approximated Bellman equations are solved sequentially with $\gamma$ values converging to 1. Theoretically, computing our new value function requires only $\frac{1}{N}$ of the time needed in prior works. }%This significantly reduces computation time.}%where $N$ bellman equations are solved sequentially, each with a corresponding value of $\gamma$ that converges to 1. In theory, the time required to compute our new value function is only $\frac{1}{N}$ of that in previous works.} %This significantly reduces computation time.} %Theoretically, our method requires only $\frac{1}{N}$ of the computation time needed by previous approaches, where $N$ is the number of times that $\gamma$ is annealed. 

%This improves the computational efficiency as the $\gamma$-annealing process in prior works is not required. %We introduce a time-discount factor $\gamma\in[0,1)$ into the RA measure. This time-discounted RA measure leads to a new value function whose Bellman equation is contractive, and its super-zero level set recovers the exact RA set, for all $\gamma \in[0,1)$. 

The central part of our new value function is the \emph{time-discounted RA measure} $\discountedreachavoidmeasure(\statetraj,t)$, for a $\gamma \in(0,1)$,\vspace{-0.2em}
\begin{equation}    \discountedreachavoidmeasure(\statetraj,t):= \min \Big\{  \gamma^t r\big(\statetrajt\big), \min_{\tau=0,\dots, t} \gamma^\tau c\big(\statetrajtau\big) \Big\}.\vspace{-0.3em}
\end{equation}
This yields a new time-discounted RA value function
\begin{equation}\label{eq:inf_horizon_reach_avoid_problem}
\begin{aligned}
    \valuefunc :=\max_{\controlpolicy} \min_{\disturbancepolicy}\sup_{t=0,\dots} \discountedreachavoidmeasure(\statetraj, t).
    % \valuefunc := %\inf_{\phi} &\max_{\mathbf{u}} 
    % \prod_{s=0}^\infty [\max_{u_s}\min_{d_s}]\sup_{t=0,\dots} &\min \big\{ \gamma^t r(\xi_{x}^{\controlseq,\disturbanceseq}(t)) ,\\& \min_{\tau=0,\dots,t} \gamma^{\tau} c(\xi_{x}^{\controlseq,\disturbanceseq}(\tau)) \big\}.
\end{aligned}\vspace{-0.3em}
\end{equation}
For all $\gamma\in(0,1)$ and any finite stage $t$, we have
\begin{equation}
    \discountedreachavoidmeasure(\statetraj,t)>0  \Longleftrightarrow \reachavoidmeasure(\statetraj,t)>0.\vspace{-0.5em}
\end{equation}
Therefore, for all $\gamma\in(0,1)$, the super-zero level set of $\valuefunc$, defined as $\superzerolevelsetofV:=\{x:\valuefunc>0\}$, is equal to the RA set $\raset$ in \eqref{eq:raset}, and it includes all possible states that can reach the target set safely in finite time under the worst-case disturbance.
\begin{figure}[t!]
    \centering
    \includegraphics[width=0.49\textwidth, trim={0cm 0.5cm 0cm 0cm}
    ]{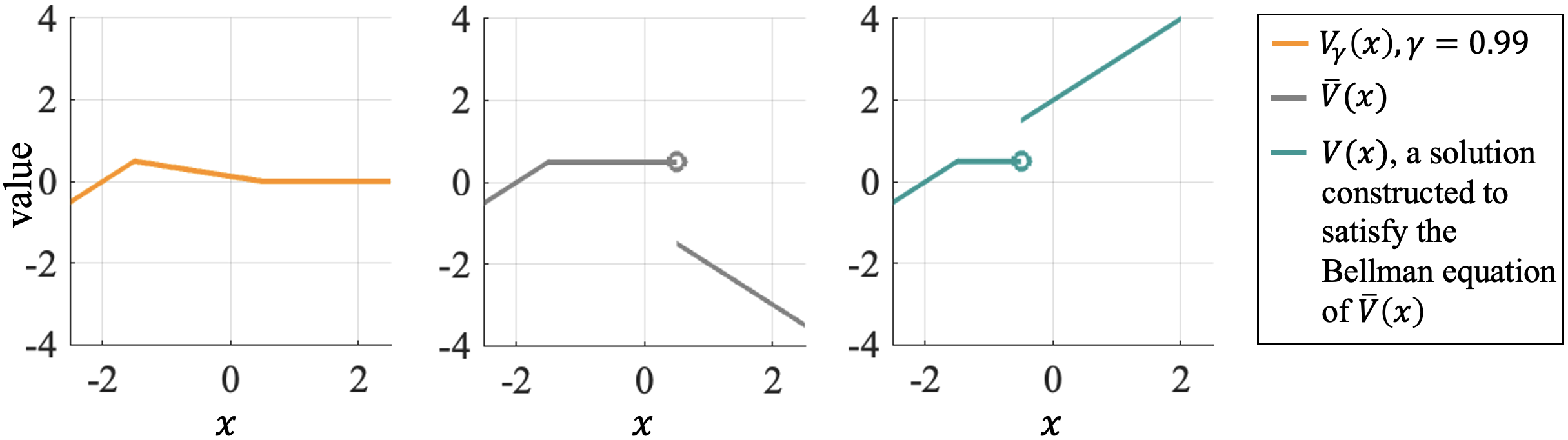}\vspace{-0.4em}
    \caption{Comparing $\valuefunc$ with $\classicalvalue(x)$ from \eqref{eq:define classical RA value} and $V(x)$, a constructed solution to the Bellman equation of $\classicalvalue(x)$ in prior works \cite{fisac2019bridging, hsusafety, hsu2023safety, hsu2023sim, hsu2023isaacs, nguyen2024gameplay, li2023learning}. Consider a 1-dimensional dynamics: $x_{t+1} = 1.01x_t + 0.01(u_t + d_t)$, with $|u_t| \le 1$ and $|d_t| \le 0.5$. We associate $\targetset = \{x : x < -1\}$ and $\constraintset = \{x : x > -2\}$ with bounded, Lipschitz continuous functions $r(x) = \max(\min( -(x + 1), 10), -10)$ and $c(x) = \max(\min(x + 2, 10), -10)$, respectively. For all $\gamma\in(0,1)$, our super-zero level set $\{x:\valuefunc>0\}$ equals the RA set $\raset=\{x:-2<x<0.5\}$. By Theorem \ref{lemma:continuity}, $\valuefunc$ is Lipschitz continuous if $\gamma\in(0, 0.99009)$. The super-zero level set of $\classicalvalue(x)$ also recovers $\raset$, but $\classicalvalue(x)$ is discontinuous at $x=0.5$ because the control fails to drive the state to $\targetset$ under the worst-case disturbance when $x_t\ge0.5$. 
    Finally, in the third subfigure, we show that the Bellman equation in prior works \cite{fisac2019bridging, hsusafety, hsu2023safety, hsu2023sim, hsu2023isaacs, nguyen2024gameplay, li2023learning} has non-unique solutions, whose super-zero level set may not equal $\raset$. %Finally, in the third subfigure, we show that the Bellman equation in prior works \cite{fisac2019bridging, hsusafety, hsu2023safety, hsu2023sim, hsu2023isaacs, nguyen2024gameplay, li2023learning} admits solutions whose super-zero level set is not equal to $\raset$. %We can construct an infinite number of such functions. %Both value functions recover the same RA set $\raset=\{x:-2<x<0.5\}$. $\valuefunc$ is Lipschitz continuous due to Theorem~\ref{lemma:continuity}, whereas $\classicalvalue(x)$ is discontinuous at $x = 0.5$.
    }\vspace{-1.4em}
    \label{fig:counter example continuity}
\end{figure}

In what follows, we present the advantages of our new value function. %this new time-discounted RA measure 
%$\discountedreachavoidmeasure$. 
First, we show that the Bellman equation for $\valuefunc$ is a contraction mapping, with $\valuefunc$ as its unique solution.\vspace{-0.3em}%, as detailed in the following theorem. %The proof can be found in the Appendix.
\begin{theorem}[Contraction mapping]\label{thm:Bellman}
Let $\gamma\in(0,1)$ and $V:\mathbb{R}^n\to\mathbb{R}$ be an arbitrary bounded function. Consider the Bellman operator $\Bellmanbackup[V] $ defined as,\vspace{-0.3em}
\begin{equation*}%\label{eq:bellman_backup}
\begin{aligned}
    \revise{\Bellmanbackup[V](x)}\hspace{-0.1cm} \coloneqq \hspace{-0.05cm} \max_u \min_{d}\min \hspace{-0.05cm} \big\{ c(x), \max\{ r(x), \gamma  V(f(x,u,d)) \}\big\}.
\end{aligned}\vspace{-0.3em}
\end{equation*}
Then, we have $    \|\Bellmanbackup[\valuefuncone] - \Bellmanbackup[\valuefunctwo]\|_\infty \le \gamma \| \valuefuncone-\valuefunctwo \|_\infty$, for all bounded functions $\valuefuncone$ and $\valuefunctwo$, and $\valuefunc$ in \eqref{eq:inf_horizon_reach_avoid_problem} is the unique solution to the Bellman equation \revise{$V(x)=\Bellmanbackup[V](x)$}. 
%\begin{equation}
%    V=B[V],\label{eq:bellman_backup_equation}
%\end{equation}
\end{theorem}\vspace{-0.5em}
\begin{proof}%\color{red}
%We first show that $\Bellmanbackup[\valuepure(x)] = \valuepure(x)$, and then prove that $\Bellmanbackup$ is a contraction mapping. 
% By definition, we have
% \begin{equation*}\small
% \begin{aligned}
% \valuepure(x_0) =& 
% \max\Big\{ \min\{ r(x_0), c(x_0) \},   \max_{\controlpolicy}\min_{\disturbancepolicy} \sup_{t=1,2,\dots} \min \{ \min\{\\&\gamma^t r(x_t), \min_{\tau=1,\dots, t} \gamma^{\tau} c(x_\tau)\}, c(x_0) \}\Big\} .
% \end{aligned}
% \end{equation*}

%Observe that we can rewrite $\sup_{t=0,\dots}\discountedreachavoidmeasure(\statetraj, t)$ as
% \begin{equation*}
%     \hspace{-1em}\begin{aligned}
%         \sup_{t=0,\dots}\discountedreachavoidmeasure(\statetraj, t) = & \max \big\{ \min\{r(x_0), c(x_0)\}, \sup_{t=1,\dots} \discountedreachavoidmeasure(\statetraj, t)\big\}\\
%         = & \max \big\{ \min\{r(x_0), c(x_0)\}, \\& \sup_{\tau=0,\dots} \min \{ \gamma \discountedreachavoidmeasure(\statetrajxnext, \tau), c(x_0)\} \big\}\\
%         = & \min \big\{ c(x_0), \max\{r(x_0), \gamma \hspace{-0.2em} \sup_{\tau=0,\dots}\hspace{-0.2em}\reachavoidmeasure(\statetrajxnext, \tau)\} \big\}
%     \end{aligned}
% \end{equation*}
\revise{Let $\controlpolicy^*$ and $\disturbancepolicy^*$ be the optimal control and the worst-case disturbance policies. Observe $\valuepure(x_0) =  \max_\controlpolicy \hspace{-0.1em}\min_\disturbancepolicy\hspace{-0.1em} \min\hspace{-0.1em}\big\{ c(x_0),  \max \{r(x_0), \hspace{-0.3em} \gamma \sup_{\tau=0,\dots}\reachavoidmeasure(\xi_{x_1}^{\controlpolicy^*, \disturbancepolicy^*}\hspace{-0.1em}, \tau)\} \big\}$ $=\max_\controlpolicy \min_\disturbancepolicy\min \big\{ c(x_0), \max \{r(x_0),  \gamma  \valuepure(x_1)\} \big\}$,
% \begin{equation*}
%     \begin{aligned}
%         \valuepure(x_0) = & \max_\controlpolicy \min_\disturbancepolicy \min\big\{ c(x_0), \max \{r(x_0),  \gamma\hspace{-0.5em} \sup_{\tau=0,\dots}\reachavoidmeasure(\xi_{x_1}^{\controlpolicy^*, \disturbancepolicy^*}, \tau)\} \big\} \\
%         = & \min \big\{ c(x_0), \max \{r(x_0), \max_\controlpolicy \min_\disturbancepolicy \gamma  \valuepure(x_1)\} \big\}
%     \end{aligned}
% \end{equation*}
where $x_1 = f(x_0,\pi(x_0),\phi(x_0, \pi(x_0)))$ is the only variable affected by $\controlpolicy$ and $\disturbancepolicy$. Following \cite[p.234]{bertsekas2012dynamic}, it can be rewritten as $\valuepure(x_0) = \max_{u_0} \min_{d_0} \min \big\{ c(x_0), \max \{r(x_0),  \gamma \valuepure(x_1)\} \big\}.$}
% \begin{equation*}
%     \valuepure(x_0) =  \min \big\{ c(x_0), \max \{r(x_0), \max_{u_0} \min_{d_0} \gamma \valuepure(x_1)\} \big\}.
% \end{equation*}

% Grouping terms with $t\in\{1,2,\dots\}$ as $\valuepure(x_1)$, we have 
% \begin{equation*}
% \begin{aligned}
% \valuepure(x_0)= &\min\Big\{ c(x_0), \max\{ r(x_0) , \gamma \max_{u_0}\min_{d_0} \valuepure(x_1) \} \Big\}.
% \end{aligned}
% \end{equation*}
% \begin{equation*}
% \begin{aligned}
% \classicalvalue(x_0)= &\min\Big\{ c(x_0), \max\{ r(x_0) ,  \max_{u_0}\min_{d_0} \classicalvalue(x_1) \} \Big\}.
% \end{aligned}
% \end{equation*}
Thus, $\valuepure(x)$ is a valid solution to the Bellman equation $V=\Bellmanbackup[V]$. We show it is a unique solution by proving that $\Bellmanbackup[V]$ is a contraction mapping \revise{when $\gamma\in(0,1)$}, i.e., $\| \Bellmanbackup[V_1] - \Bellmanbackup[V_2] \|_\infty\le \gamma \left\| V_1 - V_2 \right\|_{\infty}$, where $V_1$ and $V_2$ are two arbitrary bounded functions. Let $x$ be an arbitrary state. We have \revise{$\|\Bellmanbackup[V_1](x)  - \Bellmanbackup[V_2](x)\|_\infty
    % &\le | \max\{ r(x), \gamma \max_u\min_d V_1(f(x,u,d)) \}\\&\  \ - \max\{ r(x),\gamma \max_{u}\min_d V_2 (f(x,u,d)) \}|\\
    \le  \| \gamma \max_u \min_d V_1(f(x,u,d))  - \gamma \max_u \min_d V_2(f(x,u,d))\|_\infty$}.  %\begin{equation*}\label{eq:proof_Bellman_backup1}%\small
%    \begin{aligned}
%    &\|\Bellmanbackup[V_1(x)]  - \Bellmanbackup[V_2(x)]\|_\infty\\
%    % &\le | \max\{ r(x), \gamma \max_u\min_d V_1(f(x,u,d)) \}\\&\  \ - \max\{ r(x),\gamma \max_{u}\min_d V_2 (f(x,u,d)) \}|\\
%    \le& \| \gamma \max_u \min_d V_1(f(x,u,d))  - \gamma \max_u \min_d V_2(f(x,u,d))\|_\infty
%    \end{aligned}
%\end{equation*}
Since the max-min operator is non-expansive, we have, for all $x$, \revise{$\|\Bellmanbackup[V_1](x)  - \Bellmanbackup[V_2](x)\|_\infty\le \gamma\| V_1(x) - V_2(x)\|_\infty$}.
%where the above inequalities follow from the fact that $\|\min\{a,b\} - \min\{a,c\}\|\le b-c, \forall a,b,c\in\mathbb{R}$.
% Without loss of generality, we assume $\max_u\min_d V_1(f(x,u,d))\ge \max_u\min_d V_2(f(x,u,d))$. Let $u^*:=\argmax_u \min_d V_1(f(x,u,d))$, and $d^*:=\arg\min_d V_2(f(x,u^*,d))$. We have
% \begin{equation}\label{eq:proof_Bellman_backup2}\small
%     \begin{aligned}
%     |\gamma &\max_u\min_d V_1(f(x,u,d))-\gamma \max_u\min_d V_2 (f(x,u,d))|\\&\le \gamma |\min_d V_1( f(x,u^*,d) ) - \min_d V_2(f(x,u^*,d))|\\
%     %& \le \gamma | V_1(f(x,u^*,d^*)) - V_2(f(x,u^*,d^*)) |\\
%     & \le \gamma \max_{x} |V_1(x) - V_2(x)|
%     \le \gamma \|V_1 - V_2\|_{\infty}.
%     \end{aligned}
% \end{equation}
% From \eqref{eq:proof_Bellman_backup1} to \eqref{eq:proof_Bellman_backup2}, we have $|B[V_1](x) - B[V_2](x)|\le \gamma \| V_1 - V_2\|_\infty,\ \forall x$, 
% implying $\| B[V_1] - B[V_2] \|_\infty\le \gamma \| V_1 - V_2 \|_\infty$.
\end{proof}

%     The proof can be found in the technical report \cite{li2022infinite}. 
% \end{proof}

Theorem~\ref{thm:Bellman} suggests that annealing $\gamma$ to 1 is unnecessary in our method because, for all $\gamma\in(0,1)$, our Bellman equation admits $\valuefunc$ as the unique solution, and the super-zero level set of $\valuefunc$ equals the ground truth RA set. 
%Choosing an arbitrary $\gamma \in (0,1)$, Theorem~\ref{thm:Bellman} ensures that $\valuefunc$ is the unique solution to the Bellman equation. For all $\gamma\in (0,1)$, the super-zero level set of $\valuefunc$ equals the RA set $\raset$. Therefore, annealing $\gamma\to 1$ is unnecessary. 
% Theorem~\ref{thm:Bellman} suggests that 
% \begin{remark}
%     The Bellman equation \eqref{eq:bellman_backup} is fundamentally different from the discounted Bellman equation in \cite{fisac2019bridging, hsusafety,hsu2023safety,hsu2023sim,hsu2023isaacs,nguyen2024gameplay,li2023learning} for all $\gamma \in[0,1)$. 
% \end{remark}

% \begin{remark}
%     There are an infinite number of solutions satisfying the Bellman equation when $\gamma = 1$. This affects the learning stability when $\gamma=1$ because there is no consistent solution that the value function should converge to. Moreover, as will be shown in Section~\ref{sec:gamma experiment}, this could cause an unstable learning when $\gamma\to1$. 
% \end{remark}

% {\color{red}TODO: Introduce policy here.}

% Moreover, we present another benign property brought by the time-discount factor $\gamma \in[0,1)$. 
%We observe that the value function $\classicalvalue(x)$ could be discontinuous. This poses numerical challenges on computing the value functions. However, We show that $\valuefunc$ is Lipschitz continuous under certain conditions on the dynamics and the time-discount factor in the following theorem. %The proof can be found in the Appendix.

Furthermore, we show in the following result that our new value function can be constructed to be Lipschitz continuous, which facilitates efficient learning when approximating high-dimensional value functions using neural networks \cite{gouk2021regularisation,xiong2022deterministic}. %when $\gamma\in(0,1)$ is sufficiently small. 
% and learning \cite{gouk2021regularisation,xiong2022deterministic} %especially when we learn value functions for high dimensional systems \cite{gouk2021regularisation,xiong2022deterministic}, 
% We characterize the Lipschitz continuity in the following result.

\begin{theorem}[Lipschitz continuity]\color{black}
Suppose that $r(\cdot)$ and $c(\cdot)$ are $L_r$- and $L_c$-Lipschitz continuous functions, respectively. Assume also that the dynamics $f(x,u,d)$ is $L_f$-Lipschitz continuous in $x$, for all $u\in\mathcal{U}$ and $d\in\mathcal{D}$. Let $L:= \max(L_r,L_c)$. Then,~$\valuefunc$ is $L$-Lipschitz~continuous~if~$\gamma L_f<1$.
\label{lemma:continuity}
\end{theorem}\vspace{-0.5em}
\begin{proof}%\color{red}
Consider two arbitrary initial states $x_0,x_0'\in\mathbb{R}^n$. Let $\pi^*$ and $\phi^*$ be the optimal control and the worst-case disturbance policies. For each $t\in \{0,1,2,\dots\} $, define $x_{t+1} := f(x_t, u_t, d_t') $, $u_t: = \pi^*(x_t)$, $x_{t+1}' := f(x_t',u_t, d_t')$, and $d_t' := \phi^*(x_t', u_t)$. Let $\mathbf{u}:=\{u_t\}_{t=0}^\infty$, $\mathbf{d}':=\{d_t'\}_{t=0}^\infty$, $\mathbf{x} := \{x_t\}_{t=0}^\infty$, and $\mathbf{x}':=\{x_t'\}_{t=0}^\infty$. Given an arbitrarily small $\epsilon>0$, there exists a $\bar{t}<\infty$ such that $V_\gamma(x_0) \leq  g_\gamma(\mathbf{x},\bar{t}) + \epsilon.$ By definition, we have $V_\gamma(x_0') \geq  g_\gamma(\mathbf{x}',\bar{t}) $. Combining two inequalities, we have $V_\gamma(x_0')  - V_\gamma(x_0 ) + \epsilon 
    \geq \min\{ \gamma^{\bar{t}} (r(x'_{\bar t})-r(x_{\bar t})),
    \min_{\tau=0,...,\bar{t}} \gamma^\tau (c(x'_{\bar t}) - c(x_{\bar t})) \}
    \geq - \max\{ L_r \gamma^{\bar{t}}L_f^{\bar{t}}, \max_{\tau=0,...,\bar{t}} L_c \gamma^{\tau}L_f^{\tau } \} \|x_0-x_0'\|_2.$ 
The condition $\gamma L_f <1$ implies $(\gamma L_f)^{\bar{t}}<1$, $\forall\bar{t}$. As a result, $   V_\gamma(x'_0)  - V_\gamma(x_0 ) + \epsilon \geq  -L \|x_0-x_0'\|_2
$. Similarly, we can show that $
    V_\gamma(x_0)  - V_\gamma(x_0' ) + \epsilon \geq -L \|x_0-x_0'\|_2.
$
Combining these two inequalities, we prove Theorem \ref{lemma:continuity}.
\end{proof}
% \begin{equation}
% \begin{aligned}
%     \valuepure(x_0)  &- \valuepure(x_0' ) - \epsilon 
%     \leq  \min\{ \gamma^{\bar{t}} (r(x_{\bar t})-r(x'_{\bar t})),\\
%     &\ \ \ \ \min_{\tau=0,...,\bar{t}} \gamma^\tau (c(x_{\bar t}) - c(x'_{\bar t})) \}\\
%     &\leq\max\{ L_r \gamma^{\bar{t}}L_f^{\bar{t}}, \max_{\tau=0,...,\bar{t}} L_c \gamma^{\tau}L_f^{\tau } \} \|x_0-x_0'\|.
% \end{aligned}
% \label{eq:pf_lma_cont_eq4}
% \end{equation}

%     The proof can be found in the technical report \cite{li2022infinite}. \vspace{-0.3em}
% \end{proof}
The main idea of the proof is that a small perturbation in the state $x$ leads to a bounded change of the time-discounted RA measure value. Theorem~\ref{lemma:continuity} suggests that we can ensure the Lipschitz continuity of $\valuefunc$ by selecting $\gamma<\frac{1}{\Lf}$. In contrast, the classical RA value function $\classicalvalue(x)$ can be discontinuous, as shown in Figure~\ref{fig:counter example continuity}.

% \begin{remark}
% %In the later part of this paper, we will approximate the high-dimensional value functions using neural networks. 
% The value function's Lipschitz continuity contributes to the learning stability, as referenced in \cite{gouk2021regularisation,xiong2022deterministic}. Our experiments in Section~\ref{sec:gamma experiment} also support this. %In Section~\ref{sec:experiment}, we empirically verify that the Lipschitz continuity of $\valuefunc$ enhances the stability of the value function learning process. 
% \end{remark}

% Moreover, when $\gamma<1$, we have the trajectory with long reaching time or future risky trajectory will be depreciated. %, as shown in the following theorem.

% \begin{theorem}
%     Let $\nonanticipativestrategy$ be an arbitrary disturbance non-anticipative strategy. Suppose $g(\psi, u^1,t^1) = g(u^2,t^2)$, but $t^1<t^2$. Then, 
% \end{theorem}
% \begin{proof}
% Without loss of generality, there are two cases: either the reward dominates or the constraint dominates the RA measure.
% \end{proof}

%In our value function, higher value is preferred as a high value means we are far away from the boundary of the safe set $\constraintset$ and we are deep in the target set $\targetset$. The time-discount factor can be considered as introducing pessimism into the value function and therefore a trajectory that reaches the target set at a late stage will incur a low value due to the discounted factor. 
% The intuition behind this is that 
\revise{Moreover, the control policy derived from $\valuefunc$ reaches the target set quickly}, as a trajectory that reaches the target set rapidly incurs a high time-discounted RA measure value. %\vspace{-0.3em} %We characterize this in the following result. %, and will empirically verify it in Section \ref{sec:gamma experiment}.%We formalize this observation in the following result.
%Moreover, when we have a small $\gamma$, a safe trajectory reaching the target set at a late stage induces a low discounted RA measure value.  We show in the following result that an optimal policy maximizing the discounted RA measure tends to reach the target set faster than a one maximizing the classical RA measure. %We also present empirical results in Section \ref{sec:gamma experiment} showing that the state reaches the target set earlier when $\gamma<1$ compared to the case when $\gamma = 1$. 

\begin{theorem}[Fast reaching]\label{thm:reaching time}
    \revise{Let $x$ be in the RA set and $\disturbancepolicy$ be an arbitrary disturbance policy. Let $\gamma\in(0,1)$. Suppose $(\pione, \tone)$ and $(\pitwo, \ttwo)$ are two control policies and corresponding times to maximize the discounted RA measure $g_\gamma$, 
    % \begin{equation}
        $(\pione, \tone), (\pitwo, \ttwo)\in \arg\max_{\controlpolicy, t} g_\gamma(\xi_x^{\controlpolicy, \disturbancepolicy},t).$
    % \end{equation}
    Moreover, suppose $\tone < \ttwo$, then for all $ \check{\gamma}\in(0,\min\{\gamma,\frac{\valuefunc}{\max_x r(x)}\})$, we have $      g_{\check\gamma}(\xi_x^{\pitwo, \disturbancepolicy},\ttwo) < g_{\check\gamma}(\xi_x^{\pione, \disturbancepolicy},\tone)
$.} %\vspace{-0.5em}
    % \begin{equation}
      % g_{\check\gamma}(\xi_x^{\pitwo, \disturbancepolicy},\ttwo) < g_{\check\gamma}(\xi_x^{\pione, \disturbancepolicy},\tone)
    % \end{equation}
\end{theorem}

\begin{proof}%[Proof of Theorem~\ref{thm:reaching time}]%\renewcommand{\qedsymbol}{}
    \revise{From the definitions of $\check{\gamma}$, $\discountedreachavoidmeasure$, and $\valuefunc$, and the boundedness of $r(x)$, we have $\check{\gamma}<1$}, and $g_{\check{\gamma}}(\xi_x^{\pitwo, \disturbancepolicy}, \ttwo) \le \check{\gamma}^{\tone} \check{\gamma} \max_x r(x) \le \check{\gamma}^{\tone} V_\gamma(x) < \big(\frac{\check{\gamma}}{\gamma}\big)^{\tone} V_\gamma(x) \le \big(\frac{\check{\gamma}}{\gamma}\big)^{\tone} g_{\gamma}(\xi_x^{\pione, \disturbancepolicy}, \tone)
         \le  g_{\check{\gamma}}(\xi_x^{\pione, \disturbancepolicy}, \tone)$.
    % \begin{equation*}
    % \begin{aligned}
    %     &g_{\check{\gamma}}(\xi_x^{\pione, \disturbancepolicy}, \ttwo) \le \check{\gamma}^{\tone} \check{\gamma} \max_x r(x) \le \check{\gamma}^{\tone} V_\gamma(x) < \big(\check{\gamma}/\gamma\big)^{\tone} V_\gamma(x) \\& \le \big(\check{\gamma}/\gamma\big)^{\tone} g_{\gamma}(\xi_x^{\pione, \disturbancepolicy}, \tone)
    %      \le  g_{\check{\gamma}}(\xi_x^{\pione, \disturbancepolicy}, \tone) 
    % \end{aligned}
    % \end{equation*}
    %This completes the proof. 
\end{proof}
Theorem~\ref{thm:reaching time} also suggests that a control policy reaching the target set slowly may become suboptimal when $\gamma$ is decreased. 

While a small time-discount factor in $\valuefunc$ offers numerous benefits, it is not conclusive that $\gamma$ should always be near zero. In theory, for all $\gamma\in(0,1)$, the super-zero level set of $\valuefunc$ recovers the exact RA set. However, in practice, a near-zero $\gamma$ can lead to a conservatively estimated RA set, where a trajectory reaching the target set at a late stage may have a near-zero or even negative time-discounted RA measure due to numerical errors. We will explore the trade-offs of selecting various $\gamma$ values in Section \ref{sec:gamma experiment}. \vspace{-1em} %the time-discounted RA measure under a tiny $\gamma$ could render a trajectory reaching the target set in a long time having a near-zero, or even negative value, due to numerical error. We will characterize the trade-off of choosing $\gamma$ from different perspectives in Section \ref{sec:gamma experiment}. 

\section{Learning the New RA Value Function}\label{sec:learning}%\vspace{-0.3em}
\revise{Since the optimal RA control policy is deterministic \cite{tomlin1999computing}}, we adapt max-min Deep Deterministic Policy Gradient (DDPG) \cite{li2019robust}, a deep RL method for learning deterministic policies and their value functions, to learn $\controlpolicy$, $\disturbancepolicy$ and \revise{$\valuepure$}.% for high-dimensional nonlinear systems.% by modifying only its loss functions.% in max-min DDPG. %We assume $\gamma$ is an arbitrary time discount factor in $(0,1)$. 

Let $\gamma$ be an arbitrary time discount factor in $(0,1)$. Similar to prior works \cite{hsu2023isaacs,nguyen2024gameplay,li2023learning}, we approximate the optimal control policy $\controlpolicy^*(x)$ and the worst-case disturbance policy $\disturbancepolicy^*(x,\controlpolicy^*(x))$ by neural network (NN) policies $\NNcontrolpolicy(x)$ and $\NNdisturbancepolicy(x)$, respectively, with $\NNcontrolparam$ and $\NNdisturbanceparam$ being their parameters. We define an NN Q function as \revise{$\NNQ:\mathbb{R}^n \times \controlspace \times \disturbancespace \to \mathbb{R}$}, where $\NNQparam$ represents the NN's parameter vector. Substituting $\NNcontrolpolicy$ and $\NNdisturbancepolicy$ into $\NNQ$, we can derive an NN value function $\NNvalue(x):=\NNQ(x, \NNcontrolpolicy(x), \NNdisturbancepolicy(x))$, where $\theta$ is the concatenation of parameters $\NNQparam$, $\NNcontrolparam$ and $\NNdisturbanceparam$. Let $\mathbb{P}$ be a sampling distribution \revise{with a sufficiently large support in $\mathbb{R}^n$ that covers at least a part of the target set}. In max-min DDPG, we learn $\NNcontrolpolicy$, $\NNdisturbancepolicy$ and $\NNQ$ by \revise{alternatively} optimizing the following problems:

We learn $\NNcontrolpolicy$ by maximizing the Q value over $\NNcontrolparam$:\vspace{-0.4em}
\begin{equation}\label{eq:actor loss}
    \max_{\NNcontrolparam} \mathbb{E}_{x \sim \xdist}\  \NNQ(x,\NNcontrolpolicy(x),\NNdisturbancepolicy(x)).\vspace{-0.5em}
\end{equation}

We learn $\NNdisturbancepolicy$ by minimizing the Q value over $\NNdisturbanceparam$:\vspace{-0.4em}
\begin{equation}\label{eq:adversary loss}
    \min_{\NNdisturbanceparam}\mathbb{E}_{x\sim \xdist}\  \NNQ(x,\NNcontrolpolicy(x), \NNdisturbancepolicy(x)).\vspace{-0.5em}
\end{equation}

We learn $\NNQ$ by minimizing the \emph{critic loss}, also known as the Bellman equation error, over $\NNQparam$:\vspace{-0.4em}
\begin{equation}\label{eq:critic loss}
\begin{aligned}
    \min_\NNQparam  \mathbb{E}_{x\sim \xdist} \| \NNvalue(x) - \Bellmanbackup[\NNvalue(x)] \|_2^2.
\end{aligned}\vspace{-0.4em}
\end{equation}

\vspace{-0.5em}
% By Theorem~\ref{thm:Bellman}, we can compute $\valuefunc$ through solving
% \begin{equation}\label{eq:value loss}
%     \min_V \|V - \Bellmanbackup[V]\|_2^2.
% \end{equation}
% To make this functional space optimization problem tractable, we approximate $\valuefunc$ as a feedforward neural network (NN) $\NNvalue(\cdot):\mathbb{R}^n\to \mathbb{R}$, where $\theta$ represents the neural network's parameter vector. Let $\xdist$ be a sampling distribution of the state $x$ during training. We can reformulate the problem \eqref{eq:value loss} as
% \begin{equation}\label{eq:nn value loss}
%     \min_{\theta} \mathbb{E}_{x\sim \xdist} \|\NNvalue(x) - \Bellmanbackup[\NNvalue(x)]\|^2_2.
% \end{equation}
%Define $\NNQ(x,u,d):= \min \big\{ c(x),  \max \{ r(x), \gamma   \NNvalue(f(x,u,d)) \} \big\}$. 
% The Bellman equation in $\eqref{eq:bellman_backup}$ can be rewritten as
% \begin{equation}
%     \Bellmanbackup[\NNvalue(x)] = \max_u\min_d \NNQ(x,u,d)
% \end{equation}
% With $\NNQ$, we solve \eqref{eq:nn value loss} via max-min deep deterministic policy gradient (DDPG) \cite{li2019robust}. We learn $\NNcontrolpolicy, \NNdisturbancepolicy$ and $\NNvalue$ by jointly optimizing the following three loss functions: 

We define the \emph{learned RA set} $\learnedraset$ as the super zero-level set of $\NNvalue(x)$. 
When DDPG converges to an optimal solution%, e.g., under the conditions in \cite{xiong2022deterministic}% including the Lipschitz continuity of the value function
, $\NNvalue(x)$ converges to $\valuefunc$ due to Theorem~\ref{thm:Bellman}. 
%Though DDPG converges under certain conditions \cite{xiong2022deterministic}, 
However, in practice, like other deep RL methods, DDPG often converges to a suboptimal solution with a near-zero critic loss. %it is usually hard to verify if DDPG converges to the globally optimal solution. %DDPG can converge to a sub-optimal solution. 
When $\NNvalue(x)$ is a suboptimal solution, $\learnedraset$ cannot be reliably considered as the ground truth RA set $\raset$. 
%In the following section, we use a suboptimal control policy learned by max-min DDPG to certify a trustworthy RA set. 
This motivates us to use a suboptimal learning result to certify a trustworthy RA set, as detailed in the following section. % in the following section. 
\vspace{-1em}

\section{Certifying RA Sets with Guarantees}\label{sec:certification}%\vspace{-0.5em}
%Though ensuring the convergence of DDPG or any other deep RL method is challenging, 
In this section, we propose two methods to certify if a set of states belongs to the ground truth RA set. Both methods use a learned control policy, which is~not~necessarily~optimal.%\vspace{-1em}%In this section, we propose two methods for using a learned control policy, which is not necessarily to be optimal, to certify if a set of states belongs to the ground truth RA set. \vspace{-1em}

\subsection{Certification using Lipschitz constants}\label{sec:Lipschitz}%\vspace{-0.8em}
% {\color{blue}
% \begin{enumerate}
%     \item introduce formulation
%     \item prove the theoretical guarantee
% \end{enumerate}
% }

We leverage a learned control policy $\NNcontrolpolicy$ and the Lipschitz constants of dynamics, reward and constraint to construct a theoretical lower bound of the ground truth value function $\valuepure(x)$. \emph{If such a lower bound of $\valuepure(x)$ is greater than zero for all states in the neighboring set of $x_0$, $\neighborset_{x_0}:=\{x: \|x-x_0\|_2\le \epsilon_x\}$, then $V_\gamma(x)>0$, $\forall x\in \neighborset_{x_0}$. We claim that the set $\neighborset_{x_0}$ is within the ground truth RA set $\raset$. }

% Suppose that we can verify that a trajectory starting at state $x$ reaches the target set safely despite any disturbance within $T$ stages, then it suffices to claim $x \in \raset$. This certification horizon $T$ can be set arbitrarily, but a short horizon will lead to a conservative certification because it ignores the possibility that the trajectory may reach the target set safely~in~a~later~stage.

We begin constructing a lower bound of $\valuefunc$ by considering a $T$-stage, disturbance-free, nominal trajectory $\{\bar{x}_t\}_{t=0}^{T}$, \vspace{-0.4em}
\begin{equation}\label{eq:nominal traj}
\begin{aligned}
    \nominalx_{t+1} = f(\nominalx_t, \nominalu_t, 0),
    \nominalu_t = \NNcontrolpolicy(\nominalx_t),\ \ \forall t=0,\dots,T-1
\end{aligned}
\end{equation}
\revise{and a disturbed state trajectory $\{\nominaldisturbancetrajx_t\}_{t=0}^T$ under $\{\nominalu_t\}_{t=0}^{T-1}$ using an arbitrary $d_t\in\disturbancespace$, $\forall t=0,\dots,T-1$:}\vspace{-0.3em}
\begin{equation}\label{eq:real traj}
    \nominaldisturbancetrajx_{t+1} = f(\nominaldisturbancetrajx_t,\nominalu_t, d_t%\disturbancepolicy^*(\nominaldisturbancetrajx_t,\nominalu_t)
    ), \  \ \forall t=0,\dots,T-1. 
\end{equation}
Note that if we can verify that a trajectory starting at state $x$ reaches the target set safely \revise{despite disturbances} within $T$ stages, then it suffices to claim $x \in \raset$, where the certification horizon $T$ can be set arbitrarily. Ideally, we would set $T=\infty$, but it is impractical to evaluate an infinitely long trajectory. Therefore, during certification, we consider a finite, user-defined~$T$. This $T$ should preferably be long enough to allow initial states to reach the target set. A short $T$ results in a conservative certification since it overlooks the possibility that the trajectory might safely reach the target set at a later time.%The certification horizon $T$ can be set to any value, but a short horizon results in a conservative certification since it overlooks the possibility that the trajectory might safely reach the target set at a later time.}
%Denote by $\nominalxtraj$ a trajectory under the control trajectory $\bar{\mathbf{u}}$ and an arbitrary disturbance trajectory $\dtraj$ and an initial state $x_0 \in \neighborset_{\nominalx_0}$, i.e., $x_{t+1} =f(x_t,\nominalu_t, d_t)$. 

We assume that \revise{the dynamics $f$} is Lipschitz continuous and there exists an upper bound on the disturbance in $\disturbancespace$ at each stage $t$, i.e., $\|d_t\|_2\le \derr$. Let $\Lfx$ and $\Lfd$ be the Lipschitz constants of the dynamics $f$ with respect to the state $x$ and disturbance $d$, respectively. At time $t=1$, we observe \vspace{-0.5em}
% \begin{figure}[t!]
%     \centering
%     \includegraphics[width = 0.5\textwidth]{figure/reachable_tube.png}
%     \caption{Verifying the RA of the reachable tube under the worst case adversarial disturbance. }
%     \label{fig:enter-label}
% \end{figure}
\begin{equation*}
\begin{aligned}
    \| \nominalx_{1} &- \nominaldisturbancetrajx_{1} \|_2 {=} \|f(\nominalx_0, \nominalu_0, 0) - f(\nominaldisturbancetrajx_0, \nominalu_0, d_0%\disturbancepolicy^*(\nominaldisturbancetrajx_0,\nominalu_0) 
    )\|_2\\
    & \le \Lfx \|\nominalx_0 - \nominaldisturbancetrajx_0\|_2 + \Lfd \|0-d_0\|_2\le \Lfx \xerr + \Lfd \derr.
\end{aligned}\vspace{-0.5em}
\end{equation*}
At time $t = 2$,
\begin{equation*}
    \begin{aligned}
        \|\nominalx_2 - \nominaldisturbancetrajx_2\|_2 &{=} \|f(\nominalx_1, \nominalu_1, 0) - f(\nominaldisturbancetrajx_1, \nominalu_1, d_1)\|_2\\
        & \le \Lfx \|\nominalx_1 - \nominaldisturbancetrajx_1\|_2 +\Lfd \derr.%\\
        % & \le \Lfx\Lfx \xerr + \Lfx \Lfd \derr + \Lfd \derr
    \end{aligned}\vspace{-0.5em}
\end{equation*}
By induction, we have \vspace{-0.5em}
\begin{equation}\label{eq:xt LP bound}
    \begin{aligned}
        \|\nominalx_t - \nominaldisturbancetrajx_t\|_2\le \Lfxt \xerr + \Sigma_{\tau=0}^{t-1}\  \Lfxtau \Lfd \derr\ \  =: \Delta x_t.
    \end{aligned}
\end{equation}
We define a convex outer approximation of the set of dynamically feasible states as $\dynamicsfeasiablesetLP:=\{x_t: \|x_t - \nominalx_t\|_2 \le \Delta x_t \}$, and we check if for all $x_t\in\dynamicsfeasiablesetLP$, $r(x_t)>0$ and $c(x_t)>0$. 
By Lipschitz continuity of the reward function, we have, 
\begin{equation*}
\forall  x_t \in \dynamicsfeasiablesetLP, \ \ \ \    \|r(\bar{x}_t) - r(x_t)\|_2 \le  L_r \Delta x_t%(\Lfxt \epsilon_x + \sum_{\tau=0}^{t-1}\Lfxtau \Lfd \epsilon_d)
\end{equation*}
which yields a lower bound of $r(x_t)$, for all $ x_t \in \dynamicsfeasiablesetLP$:
\begin{equation}\label{eq:LP, reward lower}
     \rlowerLPt :=  r(\bar{x}_t) - L_r \Delta x_t \le r(x_t). %( \Lfxt \epsilon_x +  \sum_{\tau=0}^{t-1} \Lfxtau \Lfd \epsilon_d) %\underbrace{L_r( \Lfxt \epsilon_x +  \sum_{\tau=0}^{t-1} \Lfxtau \Lfd \epsilon_d)}_{\color{orange}\alpha_r(t,\epsilon_x, \epsilon_d)}
\end{equation}
Similarly, we have a lower bound of $c(x_t)$, for all $x_t \in \dynamicsfeasiablesetLP$:
\begin{equation}\label{eq:LP, constraint lower}
    \clowerLPt := c(\bar{x}_t) - L_c \Delta x_t \le c(x_t). % ( \Lfxt \epsilon_x +  \sum_{\tau=0}^{t-1} \Lfxtau L_{f_d} \epsilon_d)%\underbrace{L_c ( \Lfxt \epsilon_x +  \sum_{\tau=0}^{t-1} \Lfxtau L_{f_d} \epsilon_d) }_{\color{orange}\alpha_c(t,\epsilon_x, \epsilon_d)}
\end{equation}

\noindent Using $\rlowerLPt$ and $\clowerLPt$, we can construct a lower bound $\reachavoidcertificateL(\nominalx_0,T)$ for $\valuepure(x_0)$, for all $x_0 \in \neighborset_{\nominalx_0}$:
\begin{equation}\label{eq:Lipschitz lower bound}
\begin{aligned}
    \reachavoidcertificateL(\nominalx_0, T)& := \max_{t=0,\dots,T}\min\{ \gamma^t 
    %r\big(\xi_x^{\pi,0}(t)\big) - {\color{orange}\alpha_r(t,\epsilon_x,\epsilon_d)}
    \rlowerLPt, \min_{\tau=0,\dots,t} \gamma^\tau 
    %c(\xi_x^{\pi,0}(\tau))-{\color{orange}\alpha_c(\tau,\epsilon_x.\epsilon_d)}
    \clowerLPtau
    \}\le \valuepure(x_0).
\end{aligned}
\end{equation}
This implies %\vspace{-0.5em}%$\reachavoidcertificateL(x,T)$ is a lower bound of $\valuefunc$, 
\begin{equation}
\begin{aligned}
\reachavoidcertificateL(\nominalx_0,T)>0 \Longrightarrow &\valuepure(x_0)>0, \forall x_0 \in \mathcal{E}_{\nominalx_0}.%\\ & \forall y \in \{x: \|x-x_0\|_2\le \epsilon_x\}.    
\end{aligned}%\vspace{-0.5em}
\end{equation}
Thus, when $\reachavoidcertificateL(\nominalx_0,T)>0$, the set $\neighborset_{\nominalx_0}$ is \emph{certified} to be within the ground truth RA set $\raset$. Moreover, $\{\nominalu_t\}_{t=0}^{T-1}$ are the \emph{certified control inputs} that can drive all $x\in \neighborset_{\nominalx_0}$ to the target set $\targetset$ safely \revise{despite disturbances} in $\disturbancespace$. \vspace{-0.5em}

\subsection{Certification using second-order cone programming}
% \begin{enumerate}
%     \item Introduce a conservative surrogate problem
%     \item Formulating QPs
% \end{enumerate}
Lipschitz certification is fast to compute. However, the lower bound $\reachavoidcertificateL$ can be conservative. In this subsection, we propose another certification method using second-order cone programming (SOCP), aiming to provide a less conservative RA certification. Our key idea is to construct SOCPs that search over a tight, convex outer approximation set of the dynamically feasible trajectories and verify whether the state trajectory reaches the target set safely under all disturbances, within a user-defined finite certification horizon $T$.

%the inequality constraints in \eqref{eq:xt LP bound} are conservative because they provide a coarse over-approximation of the feasible trajectory set by using the Lipschitz constant of the dynamics. 
%In this subsection, we construct a lower bound of $\valuefunc$ using the learned control policy $\NNcontrolpolicy$ and quadratic programming. %leverage the knowledge on the dynamics model and build a tight convex over approximation of the feasible state trajectories, and then certify the constraint satisfaction and target set reaching by solving a sequence of quadratic programming problems. 

%We propose to certify whether a state can reach the taret set under a control sequence $\controlseq$ by solving QPs. 
% {
% \color{red}
% Contents design:
% \begin{enumerate}
%     \item Intuition behind QP
%     \item QP formulation
%     \item how to solve
% \end{enumerate}
% }

% The high level idea is to use QP to compute a lower bound for $c(\traj(t))$ and $r(\traj(t))$, and we claim an initial condition $x$ is certified to be able to reach the target set safely if those lower bounds are positive. 

The construction of these SOCPs involves two steps. 

First, we formulate a \emph{surrogate RA problem}: A subset of the original target set $\targetset$, represented by the interior of a polytope $\surrogatetargetset:= \{ x: P_{i} x - k_i> 0, i\in \mathcal{I}_{\surrogatetargetset} \}$, is defined as a \emph{surrogate target set} $\surrogatetargetset$, where $\mathcal{I}_\surrogatetargetset$ is a finite index set of the polytope's edges. Additionally, we define a subset of the original constraint set $\constraintset$, represented by the intersection of a finite number of \revise{positive semidefinite quadratic functions' super-zero level sets $\surrogateconstraintset:=\{x: \frac{1}{2} x^\top Q_i x + q_i^\top x + b_i > 0, i\in\mathcal{I}_{\surrogateconstraintset}\}$}, as a \emph{surrogate (nonconvex) constraint set} $\surrogateconstraintset$. %, where $\mathcal{I}_\surrogateconstraintset$ is the index set of those quadratic functions. 
The set $\surrogateconstraintset$ could be nonconvex when approximating collision avoidance constraints. %a new target set $\check{\mathcal{T}}$, represented by a convex polytope $\check{\mathcal{T}}:= \{ x: P_{i} x - k_i\ge 0, i\in \mathcal{I}_{\check{\mathcal{T}}} \}$, is defined as a subset of the original target set $\targetset$. 

%a new safe set $\check{\mathcal{S}}$, represented by the intersection of the super-zero level sets of some quadratic functions $\check{\mathcal{S}}:=\{x: \frac{1}{2} x^\top Q_i x + q_i^\top x + b_i\ge 0, Q_i\succeq0, i\in\mathcal{I}_{\check{\mathcal{S}}}\}$, is defined as a subset of the original safe set $\constraintset$. 
%Since both $\check{\mathcal{T}}$ and $\check{\mathcal{S}}$ are subsets of $\targetset$ and $\constraintset$, respectively, 
%Our intuition is that, if we can certify that a state trajectory reaches the $\check{\mathcal{T}}$ while staying in $\check{\mathcal{S}}$, then we verify that the state trajectory reaches the target set $\targetset$ safely. %$\targetset$ while being maintained within $\constraintset$. 

%The second step is solving a sequence of SOCPs at each stage $t\in \{0,1,\dots,T\}$. 
Subsequently, we leverage SOCPs to verify if a state trajectory from $x_0$ can reach $\surrogatetargetset$ while staying within $\surrogateconstraintset$ \revise{despite disturbances}. We achieve this by sequentially minimizing each function defined in $\surrogatetargetset$ and $\surrogateconstraintset$, iterating from the stage $\tau=0$ to $\tau=T$. %Since these functions are convex, they can be solved efficiently in real-time. } 
\emph{If their minimum is positive at a stage $t$ and there is no intermediate stage $\tau <t$ such that $x_\tau$ is outside $\surrogateconstraintset$, we claim that the initial state $x_0$ can safely reach the original target set $\targetset$, under all possible disturbances. }

To be more specific, we can check if there exists a stage $t\in\{0,1,\dots,T\}$ such that, for all disturbances, all dynamically feasible states $x_t$, originating from the initial states set $\mathcal{E}_{\nominalx_0}:=\{x:\|x-\nominalx_0\|_2\le \xerr\}$, are within $\surrogatetargetset$ and for all stages $\tau\le t$, $x_\tau$ are within $\surrogateconstraintset$. If this condition is met, we claim that $\mathcal{E}_{\nominalx_0}$ %the neighbor set around $\nominalx_0$ 
is within the ground truth RA set $\raset$, i.e., $\mathcal{E}_{\nominalx_0}\subseteq \mathcal{R}$. Otherwise, some of its elements could be outside $\raset$, and therefore we do not claim $\mathcal{E}_{\nominalx_0}\subseteq \raset$. %If all dynamics feasible states satisfy $x_t \in \check{\mathcal{S}}$ but there exists an $x_t \notin \check{\mathcal{T}}$, we proceed to the next stage $t+1$ and repeat the above process until we reach the certification horizon $t=T$. 
%We do not claim $\mathcal{E}_{x_0}\subseteq \raset$ for conservatism when there is no time $t\le T$ such that all dynamically feasible state $x_t$ are in the target set, $x_t \in \check{\mathcal{T}}$, and for all $\tau \le t$, all dynamically feasible state $x_\tau$ are safe, $x_\tau \in \check{\mathcal{S}}$. %If there is a state $x_t\notin\check{\mathcal{S}}$, then we do not claim $\mathcal{E}_{x_0}\subseteq \raset$ for the sake of conservativeness. We also do not claim $\mathcal{E}_{x_0}\subseteq \raset$ when there is no time $t\le T$ such that all dynamically feasible state at stage $t$ are in the target set, $x_t \in \check{\mathcal{T}}$, and for all $\tau \le t$, all dynamically feasible state at the stage $\tau$ are safe, $x_\tau \in \check{\mathcal{S}}$.
%when the following two cases happen: 1) detecting a state $x_t\notin \check{\mathcal{S}}$, $t\le T$, or 2) there exists a state $x_t\notin \check{\mathcal{T}}$ for all $t\le T$. 
% \begin{figure}[t]
%     \centering
%     \includegraphics[width=0.1\textwidth]{figure/socp.png}
%     \caption{Caption}
%     \label{fig:enter-label}
% \end{figure}
To make the analysis tractable, we define the nominal state trajectory $\{\nominalx_t\}_{t=0}^{T}$ and nominal control trajectory $\{\nominalu_t\}_{t=0}^{T-1}$ as in \eqref{eq:nominal traj}. 
% To be more specific, similar to the previous subsection, we consider a nominal state trajectory $\{\bar{x}_t\}_{t=0}^{T+1}$, which is controlled under the learned policy $\NNcontrolpolicy$ but zero disturbance,
% \begin{equation}
% \begin{aligned}
% \bar{x}_{t+1} &= f(\bar{x}_t, \nominalu_t, 0), \\ \nominalu_t &= \NNcontrolpolicy(\bar{x}_t), \forall t= 0,1,\dots,T    
% \end{aligned}
% \end{equation}
The nominal state and control trajectories allow us to formulate a convex set $\dynamicsfeasiablesetsocp$ for outer approximating the set of dynamically feasible~states:%\vspace{-0.1em}
\begin{equation*}%\label{eq:state set at time t}
\begin{aligned}
&\dynamicsfeasiablesetsocp:=\{x_t: \exists \{d_\tau\}_{\tau=0}^{t-1} \textrm{ and }x_0 \textrm{ such that }\forall \tau\le t-1,\\
&x_{\tau+1}\le \Aupper_{\tau} x_{\tau} + \Bupper_{\tau} \nominalu_\tau + \Dupper_{\tau} d_{\tau}+ \Cupper_{\tau}, \textcolor{gray}{(\textrm{Upper bound on $f$})}\\
& x_{\tau+1} \ge \Alower_{\tau} x_{\tau} + \Blower_{\tau} \nominalu_\tau + \Dlower_{\tau} d_{\tau} + \Clower_{\tau},\textcolor{gray}{(\textrm{Lower bound on $f$})}\\
& \|d_\tau\|_2\le \epsilon_d, \hspace{10.3em}\textcolor{gray}{(\textrm{Disturbance bound})}\\
& \|x_0 - \nominalx_0\|_2 \le \epsilon_x\} \hspace{8.em}\textcolor{gray}{(\textrm{Initial state bound})}
\end{aligned}
\end{equation*}
where \revise{the first two inequalities are element-wise and} the bounds on the dynamics $f$ can be derived using its Lipschitz constant or a Taylor series\footnote{\newrevise{For simplicity, we can also use $\dynamicsfeasiablesetLP$ to define the convex outer approximation set of the dynamically feasible states in SOCP certifications, i.e., $\dynamicsfeasiablesetsocp:= \dynamicsfeasiablesetLP$. %use the term $\Delta x_t$ in Lipschitz certification to define the set $\dynamicsfeasiablesetsocp$, i.e., $\dynamicsfeasiablesetsocp: = \{ x_t: \|x_t - \bar{x}_t \|_2 \le \Delta x_t\}$.
}}. At a stage $t$, we can verify if $x_t\in \surrogatetargetset$ by solving a sequence of SOCPs iterating over all $i\in \mathcal{I}_\surrogatetargetset$, and checking if their minimum is positive:\vspace{-0.5em}

\begin{equation}\label{eq:qp_certification}
    \begin{aligned}
        \rlowersocpt:=\min_{i\in \mathcal{I}_{\surrogatetargetset}} \big\{\min_{x_t \in \dynamicsfeasiablesetsocp }\  & P_i 
            x_t  -k_i \big\}. %\\
        % \textrm{s.t. } & x_t \in \stateset_t
        % &x_{t+1}  \le \Aupper x_t + \Bupper \nominalu_t + \Dupper d_t +\Cupper, \ \  \textcolor{gray}{(\textrm{Upper bound on $f$})}
        % \\
        % &x_{t+1}\ge \Alower x_t + \Blower \nominalu_t  + \Dlower d_t + \Clower, \ \  \textcolor{gray}{(\textrm{Lower bound on $f$})}
        % \\
        % &\|d_t\| \le \epsilon_d ,\  \forall t,\hspace{88pt} \textcolor{gray}{(\textrm{Disturbance bound})}\\
        % &\|x_0 - \bar{x}_0\|\le \epsilon_x. \hspace{36pt}\textcolor{gray}{(\textrm{Initial state uncertainty bound})}
    \end{aligned}\vspace{-0.5em}
\end{equation}
Similarly, for checking whether $x_t \in \surrogateconstraintset$, we can evaluate if the following term is positive:\vspace{-0.5em}
\begin{equation}
    \clowersocpt:=\min_{i\in \mathcal{I}_{\surrogateconstraintset}} \big\{ \min_{x_t \in \dynamicsfeasiablesetsocp} \frac{1}{2} x_t^\top Q_i x_t + q_i^\top x_t + b_i \big\}.\vspace{-0.5em}
\end{equation}
% Note that $\dynamicsfeasiablesetsocp$ is nonconvex when we have a non-zero, positive semidefinite $Q_i$, for some $i\in \mathcal{I}_{\check{\mathcal{S}}}$. We solve an SOCP for checking whether 
% We define 
% \begin{equation}
%     \cpolytope(\traj, t):=\min_{i\in\mathcal{I}}\min_{x_t \in \stateset_t} \constraint_i(x_t)
% \end{equation}
% as the minimum polytopic constraint value, and 
% \begin{equation}
%     \cdistancelower(\traj, t):= \min_{x_t \in \stateset_t} .
% \end{equation}
% Evaluating the above term requires solving a few quadratic programmings, i.e., we compute the minimum solution for each $i\in \mathcal{I}_j$. There always exists such lower bounds. But if we have tighter approximations, we will have less conservative RA certificate. 
% Furthermore, we have 
% \begin{equation*}
% \begin{aligned}
% \clower(\traj, t) &:= \min\{\cpolytope(\traj, t), \ \  \cdistancelower(\traj, t) \}\\
% \rlower(\traj,t) &:= \min_{x_t \in \stateset_t} \reward(x_t)
% \end{aligned}
% \end{equation*}
Combining the above two terms, we can construct a conservative certificate $\reachavoidcertificateqp$ for verifying if a state $\bar{x}_0$ and its neighboring set $\mathcal{E}_{\bar{x}_0}$ are within the ground truth RA set $\raset$,\vspace{-0.5em}
\begin{equation}\label{eq:socp lower bound}
\begin{aligned}
    \reachavoidcertificateqp(\bar{x}_0,T):= \max_{t = 0,\dots,T} \min \{\gamma^t \rlowersocpt ,  \min_{\tau=0,\dots,t} \gamma^\tau \clowersocptau \}.
\end{aligned}\vspace{-0.3em}
\end{equation}
This suggests\vspace{-0.5em}
\begin{equation}
    \begin{aligned}
        \reachavoidcertificateqp(\nominalx_0,T)>0 \Longrightarrow \valuepure(x_0)>0,\forall x_0\in \mathcal{E}_{\nominalx_0}
    \end{aligned}
\end{equation}
and $\{\nominalu_t\}_{t=0}^{T-1}$ are the \emph{certified control inputs}, capable of driving all $x\in \neighborset_{\nominalx_0}$ to the target set $\targetset$ safely \revise{despite disturbances}. 

\noindent \textbf{Running example (continued).} At stage $t$, we evaluate $\rlowersocpt = \min_{i}  \{\check{r}_{t,i}^{S}\}_{i=1}^6$ by solving the following SOCPs, where each $\check{r}_{t,i}^{S}$ corresponds a function in the definition of $\targetset$ in \eqref{eq:drone target}:\vspace{-0.5em}
\begin{equation*}
\begin{aligned}
&\check{r}_{t,1}^{S}=\min_{x_t\in \dynamicsfeasiablesetsocp} p_{y,t}^1 - p_{y,t}^2, 
&&\check{r}_{t,2}^{S}=\min_{x_t\in \dynamicsfeasiablesetsocp}  v_{y,t}^1 - v_{y,t}^2 \\
&\check{r}_{t,3}^{S}=\min_{x_t\in \dynamicsfeasiablesetsocp}0.3 - p_{x,t}^1 ,  &&\check{r}_{t,4}^{S}= 
\min_{x_t\in \dynamicsfeasiablesetsocp}  0.3 +p_{x,t}^1 \\
& \check{r}_{t,5}^{S}= \min_{x_t\in \dynamicsfeasiablesetsocp} 0.3 - p_{z,t}^1, &&\check{r}_{t,6}^{S}=\min_{x_t\in \dynamicsfeasiablesetsocp} 0.3 + p_{z,t}^1
\end{aligned}\vspace{-0.5em}
\end{equation*}
Similarly, we can evaluate $\clowersocpt=\min_i \{\check{c}_{t,i}^{S}\}_{i=1}^5$ by considering the following SOCPs, where each $\check{c}_{t,i}^{S}$ corresponds to a constraint function in \eqref{eq:constraint pass} and \eqref{eq:safety cone}:\vspace{-0.5em}
\begin{equation*}
    \begin{aligned}
        &\check{c}_{t,i}^{S}= \min_{x_t\in\dynamicsfeasiablesetsocp}   (-1)^i\times p_{x,t}^1 - p_{y,t}^1+0.05 ,  \hspace{0.2cm} i\in\{1,2\},\\  %\check{c}_{t,2}^{S}=\hspace{-0.2cm}\min_{x_t\in \dynamicsfeasiablesetsocp}\hspace{-0.2cm} -p_{x,t}^1 - p_{y,t}^1 + 0.05,\\
        & \check{c}_{t,i}^{S} = \min_{x_t\in \dynamicsfeasiablesetsocp}  (-1)^i\times p_{z,t}^1 - p_{y,t}^1 + 0.05, \hspace{0.2cm} i\in\{3,4\}. \\ %\check{c}_{t,4}^{S} = \hspace{-0.2cm} \min_{x_t\in \dynamicsfeasiablesetsocp}\hspace{-0.2cm} -p_{z,t}^1 - p_{y,t}^1 + 0.05,\\
        %& p_{z,t}^{1,\textrm{min}}= \max_{x_t \in \dynamicsfeasiablesetsocp} p_{z,t}^1 ,\ \  p_{z,t}^{2,\textrm{max}}= \min_{x_t \in \dynamicsfeasiablesetsocp} p_{z,t}^2 \\
    \end{aligned}\vspace{-0.5em}
\end{equation*}
We overapproximate the maximum height difference between two drones via $\Delta_{z,t}^{21}:=\max_{x_t\in \dynamicsfeasiablesetsocp}p_{z,t}^2 - p_{z,t}^1$, and consider \vspace{-0.5em}
\begin{equation*}
        \check{c}_{t,5}^{S} = \hspace{-0.2cm}\min_{x_t\in \dynamicsfeasiablesetsocp} \left\|\begin{bmatrix}
            p_{x,t}^1 - p_{x,t}^2 \\ p_{y,t}^1 - p_{y,t}^2
        \end{bmatrix}\right\|_2^2 - (1+\max(\Delta_{z,t}^{21},0))\times 0.2 
\end{equation*}%\vspace{-1.4em}
\begin{remark}%[Methods comparison]%We note that both methods can be adapted to certify RA sets for prior reachability~learning~works~\cite{li2023learning, nguyen2024gameplay}. 
%In general, SOCP certification can be less conservative than Lipschit certification, but the latter can be faster to compute: 
1)
SOCP certification can be less conservative than Lipschitz certification when $\surrogatetargetset$ and $\surrogateconstraintset$ can represent $\targetset$ and $\constraintset$ exactly, and $\dynamicsfeasiablesetsocp$ is a subset of the Lipschitz dynamically feasible set $\dynamicsfeasiablesetLP$, as defined in Section~\ref{sec:Lipschitz}, for all $t\le T$. %the original target set $\mathcal{T}$ and constraint set $\constraintset$ are polytopes, 
This is because Lipschitz certification adds extra conservatism when estimating lower bounds of $r(x)$ and $c(x)$ using the Lipschitz constant, as shown in \eqref{eq:LP, reward lower} and \eqref{eq:LP, constraint lower}; 2) However, Lipschitz certification is faster to compute than SOCP because calculating $\rlowerLPt$ and $\clowerLPt$ is easier than~evaluating~$\rlowersocpt$~and~$\clowersocpt$; \revise{3) SOCP certification employs convex over-approximations of the dynamically feasible state sets, similar to tube MPC \cite{lopez2019tubempc}. However, it differs from tube MPC in that we compute the worst-case (nonconvex) constraint violation and target set deviation \newrevise{rather than the control inputs that optimize an objective.}% instead of the optimal control.%by evaluating the worst-case constraint violation and target set deviation instead of the optimal control. 
}%; 3) The computational complexity of evaluating $\reachavoidcertificateL(x,T)$ and $\reachavoidcertificateqp(x,T)$ both \textbf{scales polynomially with the dimension of the dynamical system} and the length of $T$, but it is not strongly dependent on the size of $\neighborset_{x}$. 
\end{remark}\vspace{-1em}
\begin{remark}
    The computational complexity of evaluating $\reachavoidcertificateL(x,T)$ and $\reachavoidcertificateqp(x,T)$ scales \textbf{polynomially} with both the dimension of the dynamical system and the length of $T$.%, but it is not strongly dependent on the size of $\neighborset_{x}$. 
\end{remark}

%\vspace{-1.5em}
% because the later requries to solving a sequence of convex programmings. %Lipschitz certification applies to nonconvex target sets and constraint sets, but the SOCP certification requires to first derive polytopic subsets of the target and safe sets and then certify whether we can reach the surrogate target set while being maintained within the surrogate polytopic safe set. %The only requirement is the Lipschitz contiunity of the dynamics, reward and constraints. 
%The SOCP certification does not require the safe set and the target set to be polytopes. When we have nonconvex sets, we can do inner convex approximation of those sets but the price is conservatism. %However, SOCP certification is more computational demanding than Lipschitz certification due to the need of solving SOCPs. % because the number of SOCPs to be solved scales linearly with the length of the certification horizon $T$. 
%Both methods could be computed in real-time for a single state $x_0$ and its neighbor set $\mathcal{E}_{x_0}$, as we will show in Section \ref{sec:experiment}. 

% \begin{remark}
%     Our certification techniques do not restrict $\NNcontrolpolicy$ to be learned using our new value function. They can also be applied to certify RA sets for prior works \cite{li2023learning, nguyen2024gameplay}. %However, a poorly learned policy can lead to a poor quality of the certified RA~set.
%     \vspace{-0.7em}
% \end{remark}

\section{Combining reachability learning and certification}%\vspace{-0.3em}
%\subsection{Offline and online certifications}
We integrate the reachability learning and certification into a new framework of computing trustworthy RA sets, as described in Algorithm~\ref{alg:calibration}. %As described in Algorithm~\ref{alg:calibration}, we first apply max-min DDPG to learn the RA value function and its policies. Then, we can conduct certifications and 
%We are guaranteed to complete the RA task, if we certify that the initial state is within the ground truth RA set:
% {\color{orange}We remark that the control policy used for certification does not have to be learned using our new value function. Our certification methods can also be applied to certify RA sets for prior works \cite{li2023learning, nguyen2024gameplay}. }
The super zero-level set of $\NNvalue$ provides an estimation of the ground truth RA set. We use $\NNcontrolpolicy$ to certify a set of states, ensuring deterministic RA guarantees there. In particular, we can apply certification either online or offline:

\noindent\textbf{Online certification:} Let $x$ be an arbitrary state. We can use the RA certificates $\reachavoidcertificateL(x,T)$ in \eqref{eq:Lipschitz lower bound} or $\reachavoidcertificateqp(x,T)$ in \eqref{eq:socp lower bound} as online RA certification methods, verifying if all elements in $\neighborset_x=\{x':\|x'-x\|_2\le\xerr\}$ can reach the target set safely \revise{despite disturbances}. We can compute $\reachavoidcertificateL(x,T)$ and $\reachavoidcertificateqp(x,T)$ in real-time (10 Hz \newrevise{in our examples}), as shown in Figure~\ref{fig:histogram}. \revise{It can also be integrated into the safety filter proposed in \cite{hsu2023isaacs}, offering more robust RA certification than \cite{hsu2023isaacs} by verifying that all states in \(\neighborset_x\) can reach the target set safely, thereby enabling RA capability verification without perfect state estimation.}
% It can also be integrated into the safety filter proposed in \cite{hsu2023isaacs}, and it provides a more robust RA certification than \cite{hsu2023isaacs} as we certify whether all states in $\neighborset_x$ can reach the target set safely, allowing us to verify the RA capability without perfect state estimation. }

\noindent\textbf{Offline certification:} 
We consider a finite set of states $\mathcal{L}:=\{x^{(i)}\}_{i=1}^N$ such that the union of their neighboring sets $\neighborset_{x^{(i)}}=\{x:\|x-x^{(i)}\|_2\le \xerr\}$ covers the set of states that we aim at certifying. For example, this includes the area near the orange gate in drone racing, as show in Figure~\ref{fig:front}. We enumerate each element $x\in \mathcal{L}$ and certify whether $\mathcal{E}_{x}\subseteq \raset$ by checking if $\reachavoidcertificateL$ in \eqref{eq:Lipschitz lower bound} or $\reachavoidcertificateqp$ in \eqref{eq:socp lower bound} is positive. The union $\certifiedset$ of those certified sets constitutes a subset of the ground truth RA set $\raset$. For all elements in $\certifiedset$, we guarantee that they can reach the target set safely \revise{under} potential disturbances. %\vspace{-0.7em}%We remark that offline certification requires more computational resources than online certification, but it results in larger certified RA sets. 

% \vspace{-0.1em}

% \begin{remark}[Certified control]
    
% \end{remark}

%Moreover, this online certification can be done in real-time in drone racing hardware experiments in Section \ref{sec:experiment}.

% We can efficiently evaluate the two construc

%The conservatism of the RA certificates, e.g., $\reachavoidcertificateL$ and $\reachavoidcertificateqp$, could be further relaxed if we have prior knowledge of the worst-case disturbances. %We leave future research. %Our method paves the way for certifying trust-worthy RA set for high-dimensional nonlinear systems. %We remark that we can build a tighter lower bound for the RA value function if we know the worst-case disturbance policies. We leave the construction of less conservative RA certificate for future research. 

\begin{algorithm}[t]
	\SetAlgoLined
	\caption{Certifiable Reachability Learning:}\label{alg:calibration}\tikzmark{left}
	\hspace{-0.5em}\textbf{Require:} \ an arbitrary $\gamma\in(0,1)$, a finite list of states $\mathcal{L}$, and a certification horizon $T<\infty$\;%a state set $\mathcal{X}^*$ that is to be certified, and a collection of states $\mathcal{L}:=\{x^{(1)}, \dots,x^{(N)} \}$ such that $\cup_{i=1}^N \mathcal{E}_{x^{(i)}}\supseteq \mathcal{X}^*$\;
    \textbf{Initialization:} certified RA set $\certifiedset\gets\{\}$\;
    %\tcp{Offline reachability learning:}
    Learn $\NNcontrolpolicy$, $\NNdisturbancepolicy$ and $\NNvalue$ via max-min DDPG \cite{li2019robust}\;
    % \If{Offline certification}
    % {
    %\tikzmark{begin}
    % \hspace{-0.5em}
    % Obtain a finite list of initial states $\userdefinedset$\;
    \tcp{Certifications:}
    \For{$x_0 \in \mathcal{L}$ }
    {
        \If {$\reachavoidcertificateL(x_0,T)>0$ \emph{or} $\reachavoidcertificateqp(x_0,T)>0$} {$\certifiedset\gets \certifiedset\cup \{x:\|x_0-x\|_2\le \xerr\}$}
    }
    %\tikzmark{end}\drawredbox
    % }
    % \If{Online certification}{\tikzmark{begin}
    % \hspace{-0.5em}Obtain the current state $x_0$\;
    % \If {$\check{V}(x_0,T)>0$} {$\certifiedset\gets \certifiedset\cup\{y:\|x_0-y\|_2\le \epsilon\}$}\tikzmark{end}\drawredbox
    % }
    % Obtain a certified RA set $\certifiedset\subseteq \userdefinedset$\;
    % $\mathcal{P}\gets \bigcup \textrm{traj}_{\pi_\theta }(\mathcal{S})$ (then, $\mathcal{P}\subset True RA set$) \\
    \algorithmicreturn \  certified RA set $\certifiedset$
\end{algorithm}%\vspace{-2em}

\section{Experiments}\label{sec:experiment}%\vspace{-0.2em}
We test our reachability learning and certification methods\footnote{%\color{red}
Experiment code and hardware drone racing video are available at \href{https://github.com/jamesjingqili/Lipschitz_Continuous_Reachability_Learning.git}{https://\\github.com/jamesjingqili/Lipschitz\_Continuous\_Reachability\_Learning.git}.} in a 12-dimensional drone racing hardware experiment (Figure~\ref{fig:hardware experiment}), and a triple-vehicle highway take-over simulation. In the highway simulation, we control one ego vehicle, modeled with nonlinear unicycle dynamics, to safely overtake another vehicle while avoiding a third vehicle driving in the opposite direction, as shown in Figure \ref{fig:highway_sets}. Figures \ref{fig:front} and \ref{fig:highway_sets} demonstrate the high quality of the learned and certified RA sets.%\vspace{-1em} %We will discuss the details of our experiments in the following subsections. 
%Additionally, the hardware experiments in Figure~\ref{fig:hardware experiment} show that our learned control policy $\NNcontrolpolicy$ safely controls the ego drone (highlighted with brightness) to overtake the other drone when the initial states of both agents are within the SOCP certified set.\vspace{-1em}

% Specifically, we investigate the following three hypotheses:
% \begin{enumerate}
%     \item Our learned RA policy has a high success rate than the state-of-the-art
% \end{enumerate}

\subsection{Hypothesis 1: Our learned policy has a higher success rate than state-of-the-art constrained RL methods}%\vspace{-0.5em}
We compare our learned policy $\NNcontrolpolicy$ with Deep Deterministic Policy Gradient-Lagrangian (DDPG-L)\cite{chow2019lyapunov}, Soft Actor Critic-Lagrangian (SAC-L) \cite{yang2021wcsac}, and Constrained Policy Optimization (CPO) \cite{achiam2017constrained}. %, and therefore we could not report CPO result for drone racing. 
We summarize the results in Table I. The success rate is estimated by computing the ratio of sampled initial states that can reach the target set safely under randomly generated disturbances. Our method achieves a 1.0 success rate when the initial states are sampled from the SOCP certified set, validating the deterministic guarantee that all elements in the certified sets can safely reach the target set \revise{despite disturbances}. %However, CPO fails to converge for the drone racing experiment due to the complex and nonconvex constraints. \vspace{-1em}

\begin{table}[t!]
\centering
% Reduce font size
\small

% Reduce space between columns
\setlength{\tabcolsep}{3pt}

\begin{adjustbox}{max width=\linewidth}
\begin{tabular}{l|ccc|cccc}
\toprule
\multirow{2}{*}{\begin{tabular}[c]{@{}l@{}}Initial states \\ uniform sampling\end{tabular}} 
        & \multicolumn{3}{c|}{Drone racing} 
        & \multicolumn{4}{c}{Highway} \\
\cmidrule(lr){2-4}\cmidrule(lr){5-8}
     & DDPG-L & SAC-L & Ours 
     & DDPG-L & SAC-L & CPO & Ours \\
\midrule
In a large bounded state set    & 0.5716 & 0.7291 & \textbf{0.7655} & 0.7512 & 0.6343 & 0.5552 & \textbf{0.8782} \\
In the learned reach-avoid set  & 0.6276 & 0.8006 & \textbf{0.8889} & 0.9430 & 0.9149 & 0.9566 & \textbf{0.9924} \\
In the SOCP certified set       & 0.9673 & 0.9948 & \textbf{1.0000} & 0.9111 & 0.8850 & 0.9451 & \textbf{1.0000} \\
\bottomrule
\end{tabular}
\end{adjustbox}%\vspace{-0.5em}
\caption{Success rates table. Our method achieves a $1.0$ success rate when the initial states are sampled from the SOCP certified~set. CPO fails to converge for the drone racing experiment due to the complex and nonconvex constraints.}\label{table:success rate}\vspace{-1em}
\end{table}

\begin{figure}[t!]%\vspace{-0.5em}
    \centering
    \includegraphics[width=0.49\textwidth, trim={0cm 0.3cm 0cm 0cm}]{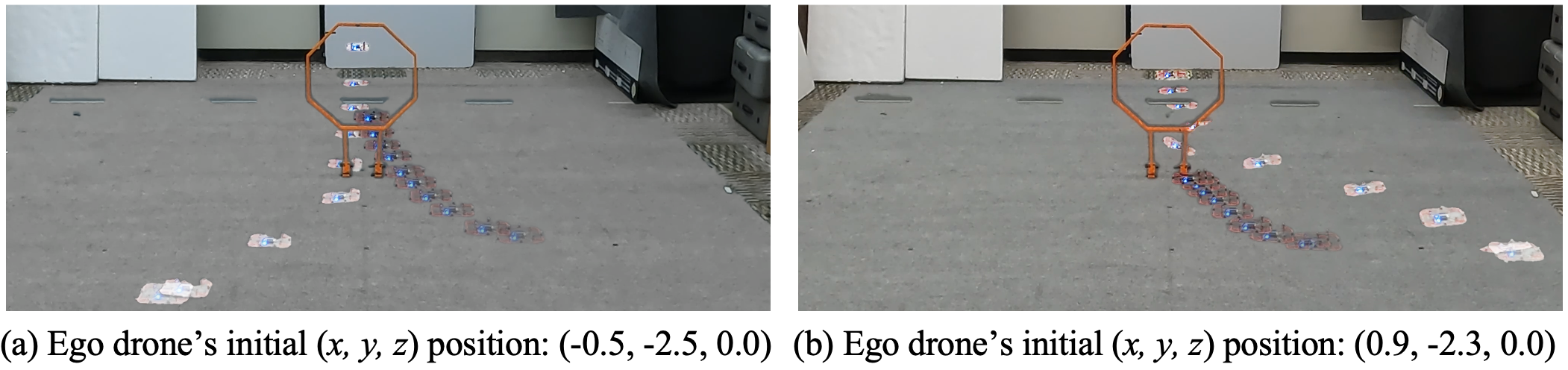}\vspace{-0.3em}
    \caption{\revise{We sampled 50 initial states from the SOCP certified set shown in Figure~\ref{fig:front}. A few crashes occurred due to insufficient battery charge or Vicon sensor failures caused by natural light. These instances were excluded as outliers. With a fully charged battery and no Vicon system failures, the ego drone successfully overtook the other drone from each of the 50 initial states, despite the latter’s uncertain acceleration.%With full battery and no Vicon systems failure, the ego drone safely overtakes the other drone from each of those 50 initial states, despite the uncertain acceleration of the other drone.
    } We visualize two hardware experiments in the above subfigures. The remaining 9-dimensional initial state includes $[v_{x,t}^1,v_{y,t}^1,v_{z,t}^1, p_{x,t}^2,p_{y,t}^2,p_{z,t}^2, v_{x,t}^2, v_{y,t}^2, v_{z,t}^2]=[0,0.7, 0,0.4,-2.2,0,0,0.3,0]$.  %The second drone's initial $(x,y,z)$ position is set to be $(0.4, -2.2, 0.0)$. Both drones have zero initial $x$- and $z$-velocities. The $y$-velocities of the ego drone and the other one are $0.7$ and $0.3$, respectively, 
    }\vspace{-1.em}
    \label{fig:hardware experiment}
\end{figure}

\begin{figure}[t!]\vspace{-0.2em}
    \centering
    \includegraphics[width = 0.49\textwidth,trim={0cm 0.5cm 0cm 0cm}]{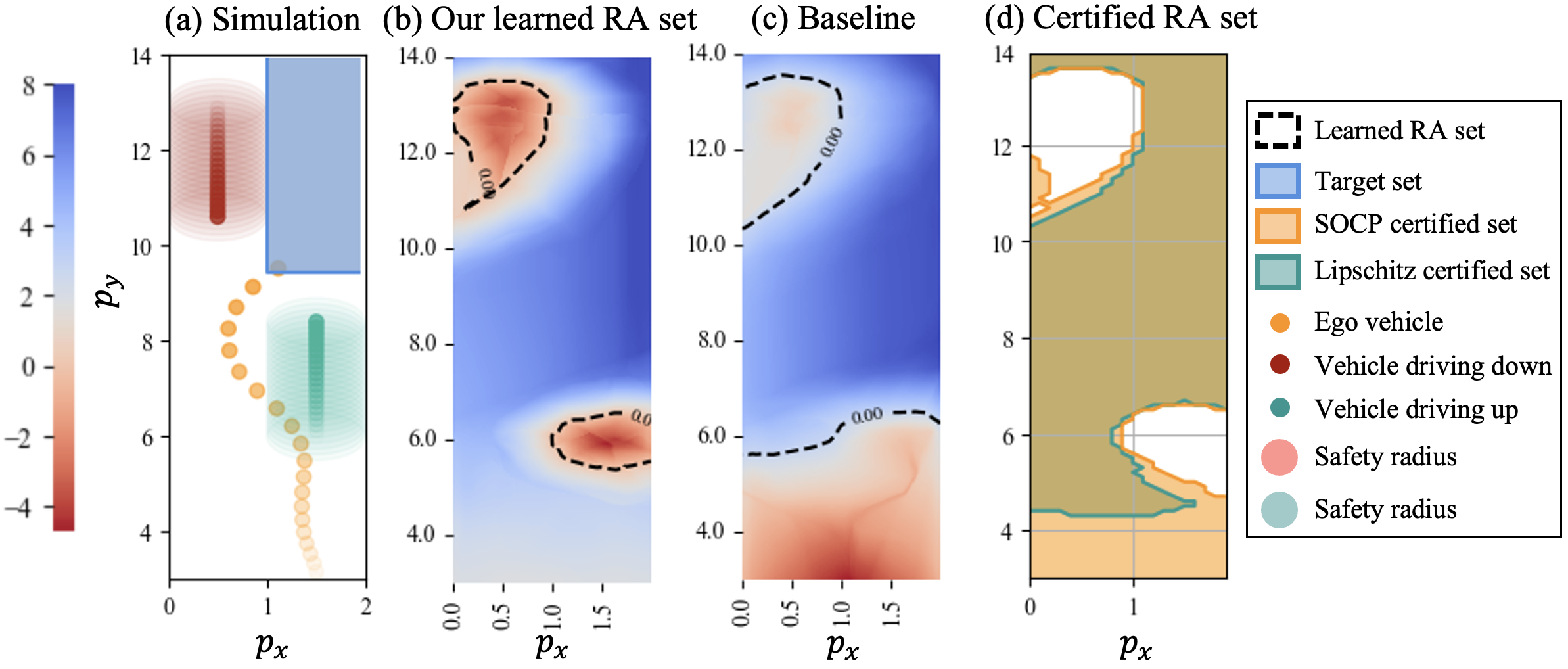}\vspace{-0.3em}
    \caption{Highway reachability analysis: In (a), we simulate the nonlinear dynamics with the learned policy $\NNcontrolpolicy$ and randomly sampled disturbances on other vehicles' acceleration. The 10-dimensional state space includes $[p_{x,t}^1, p_{y,t}^1,v_{t}^1,\theta_t^1, p_{x,t}^2,p_{y,t}^2, v_{y,t}^2, p_{x,t}^3,p_{y,t}^3, v_{y,t}^3]$. {%\color{red}
    The $p_y$-axis movement of the red and green agents is modeled using double integrator dynamics, while their initial $p_x$ positions are sampled randomly and remain stationary during simulation}. In (b), we project our learned value function, with $\gamma=0.95$, onto the $(x,y)$ position of the ego vehicle. In (c), we plot the RA set learned using the state-of-the-art method \cite{li2023learning,nguyen2024gameplay} with $\gamma=0.95$. As suggested in \cite{hsusafety}, annealing $\gamma\to 1$ is necessary for prior works; otherwise, the learned RA sets in prior works are conservative. In (d), we plot our certified RA sets. }\vspace{-1.3em}
    \label{fig:highway_sets}
\end{figure}

% \begin{figure}[t!]
%     \centering
%     \includegraphics[width = 0.5\textwidth]{figure/drone_racing_sets.png}
%     \caption{Drone racing reachability analysis: visualization of the learned RA set, Lipschitz and SOCP certified sets, and a simulation of our policy against sampled disturbances}
%     \label{fig:drone_racing_sets}
% \end{figure}

\begin{figure}[t!]
    \centering
    \includegraphics[width = 0.49\textwidth, trim={0cm 0.7cm 0cm 0cm}]{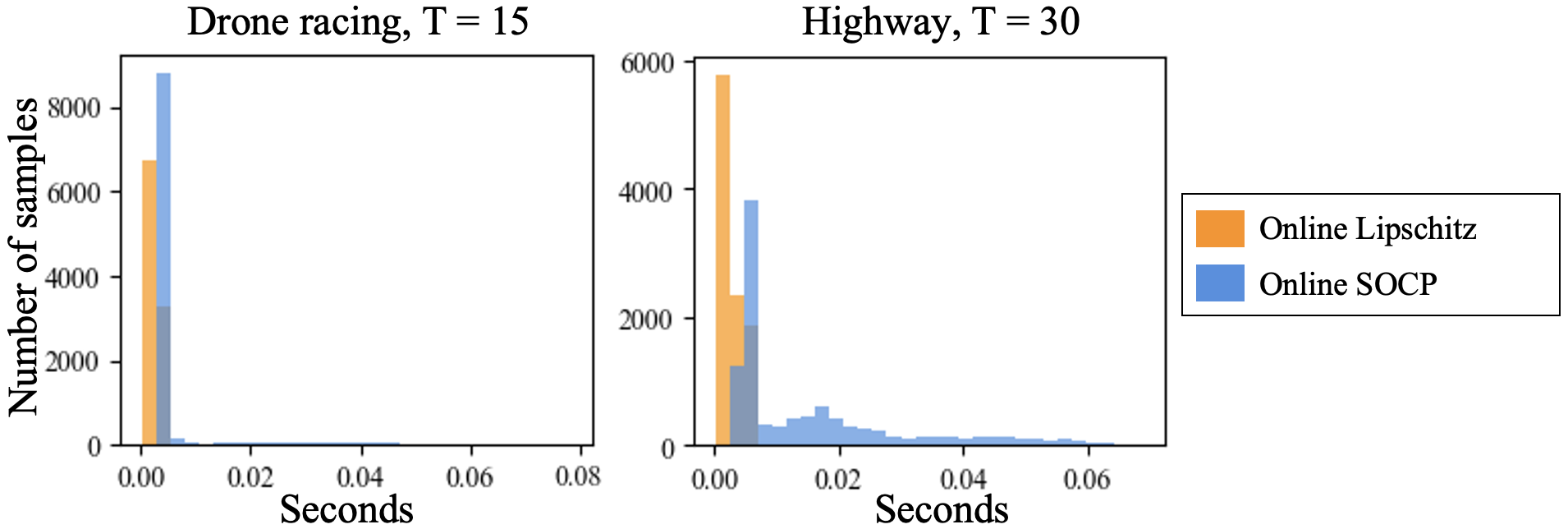}
    \caption{Histogram of the time required for computing $\reachavoidcertificateL(x,T)$ and $\reachavoidcertificateqp(x,T)$ for each of the 10,000 randomly sampled states $x$. The certification horizons for drone racing and highway are $T=15$ and $T=30$, respectively. 
    %local certifications of whether all elements in the neighbor set around $x_t$ can reach the target set safely despite any disturbances. 
    }%\vspace{-1.5em}
    \label{fig:histogram}
\end{figure}

\subsection{Hypothesis 2: Our online RA set certification methods can be computed in real-time}%\vspace{-0.2em}
Figure \ref{fig:histogram} shows that $\reachavoidcertificateL(x,T)$ and $\reachavoidcertificateqp(x,T)$ can be computed in real-time to certify if all elements in $\neighborset_{x}=\{x':\|x'-x\|_2\le 0.1\}$ can safely reach the target set, \revise{under} all potential disturbances. This enables real-time online certification.\vspace{-1em} %With a properly chosen certification horizon $T$, we could compute $\reachavoidcertificateL$ and $\reachavoidcertificateqp$ in real-time. 
% The computational complexity of evaluating $\reachavoidcertificateL(x,T)$ and $\reachavoidcertificateqp(x,T)$ \textbf{scales polynomially with the dimension of the dynamical system} and the length of $T$, but it is not strongly dependent on the size of $\neighborset_{x}$. %A short $T$ may lead to conservatism by overlooking the possibility that the state trajectory reaches the target set safely at a later stage. 
% \vspace{-1em}
%When $T$ is small, the RA certificates value could be conservative because it overlooks the possibility that the state trajectory could reach the target set safely in a late stage. However, when $T$ is large, then 

% \begin{figure}[t!]
%     \centering
%     \includegraphics[width=0.4\textwidth]{figure/droneracing_traj.png}
%     \caption{Comparing trajectories under different policies in drone racing. }
%     \label{fig:enter-label}
% \end{figure}
% \begin{figure}[t!]
%     \centering
%     \includegraphics[width=0.35\textwidth]{figure/highway_traj.png}
%     \caption{Comparing trajectories under different policies in highway take-over. }
%     \label{fig:enter-label}
% \end{figure}
% {\color{red} arrow, indicating the direction. }

\begin{figure}[t!]
    \centering
    \includegraphics[width=0.49\textwidth, trim={0cm 0.5cm 0cm 0cm}]{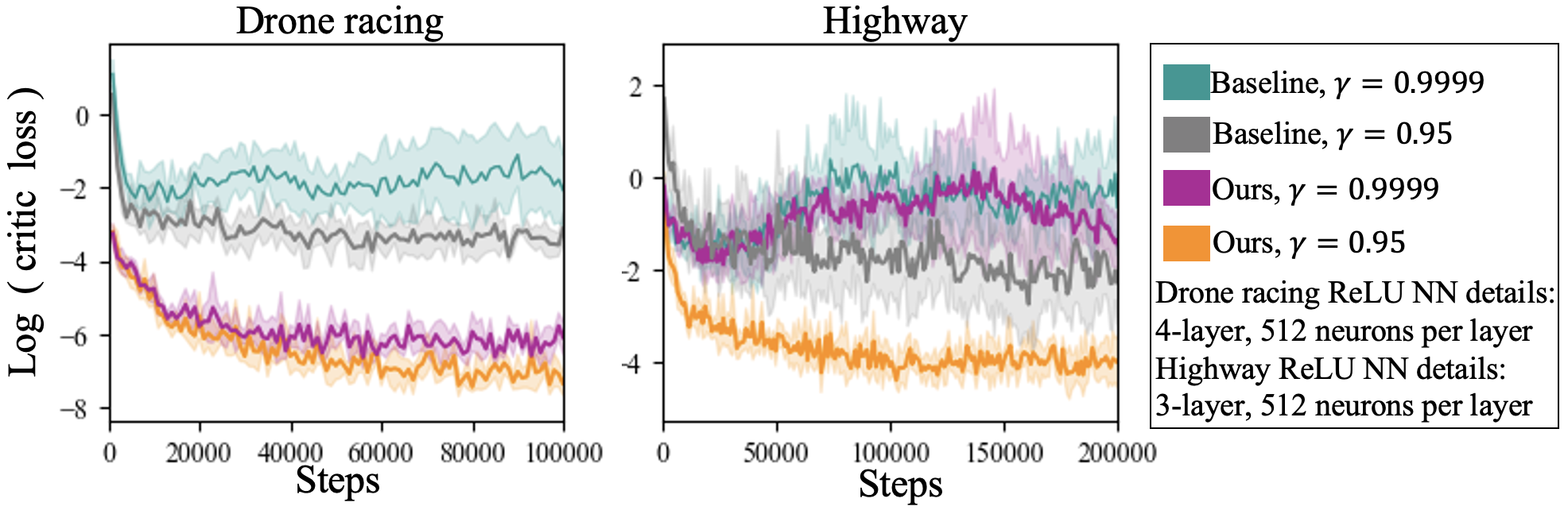}\vspace{-0.4em}
    \caption{
    %Comparing the critic loss of max-min DDPG under our new Bellman equation (Ours) and under the approximated Bellman equation (Baseline) from prior works \cite{nguyen2024gameplay, li2023learning}, with different $\gamma$ values. 
    We compare the convergence of the critic loss under our new Bellman equation with the baseline from previous works \cite{nguyen2024gameplay, li2023learning}, using different $\gamma$ values but identical training parameters. Our critic loss with $\gamma=0.95$ converges faster than with $\gamma=0.9999$, likely due to the Lipschitz continuity of $\valuefunc$ at $\gamma=0.95$. \revise{The training speed is around 1700 steps per second. }%The convergence of the critic loss under our new Bellman equation (Ours) and under the approximated Bellman equation (Baseline) from prior works \cite{nguyen2024gameplay, li2023learning}, with different $\gamma$ values but same training parameters. Our critic loss under $\gamma=0.95$ converges more rapidly than under $\gamma=0.9999$ possibly due to the Lipschitz continuity of $\valuefunc$ when $\gamma=0.95$. %The baseline method uses the time-discounted Bellman backup approximation in \cite{li2023learning,nguyen2024gameplay}. 
    %Both the baseline and ours are trained using the same training parameters and random initializations. For drone racing, we parameterize each of $\NNcontrolpolicy$, $\NNdisturbancepolicy$ and $\NNQ$ using a 4-layer ReLU NN, with 512 neurons in each layer. For the highway example, we parameterize each $\NNcontrolpolicy$, $\NNdisturbancepolicy$ and $\NNQ$ using a 3-layer ReLU NN, with 512 neurons in each layer. We conduct NN training under 5 different random~seeds. %In the above experiments, we fix the model parameters and the only difference is the Bellman backup. 
    %Since our newly designed value function is Lipschitz continuous and its ground truth Bellman equation is a contraction mapping, without the need of any approximation, commonly used in prior works \cite{fisac2019bridging, hsusafety, hsu2023safety, hsu2023sim, hsu2023isaacs, nguyen2024gameplay, li2023learning}. 
    }\vspace{-1em}
    \label{fig:learning 1}
\end{figure}
\begin{figure}[t!]\vspace{-0.3em}
    \centering
    \includegraphics[width=0.49\textwidth, trim={0cm 0.5cm 0cm 0cm}]{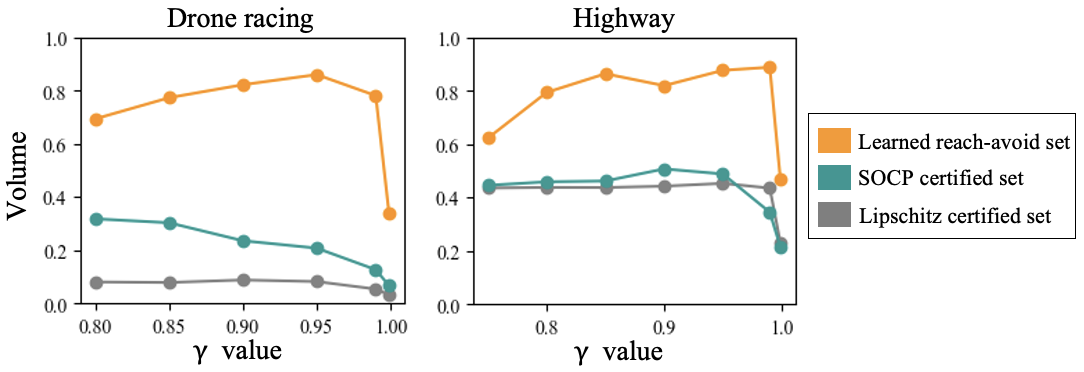}
    \caption{The volumes of the learned RA set, SOCP certified set, and the Lipschitz certified set change as $\gamma$ varies. We estimate the set volumes using the Monte Carlo method with 10,000 random samples in the state space. }\vspace{-1.1em}
    \label{fig:gamma effect 1}
\end{figure}
\begin{figure}[t!]\vspace{-0.3em}
    \centering
    \includegraphics[width=0.49\textwidth, trim={0cm 0.5cm 0cm 0cm}]{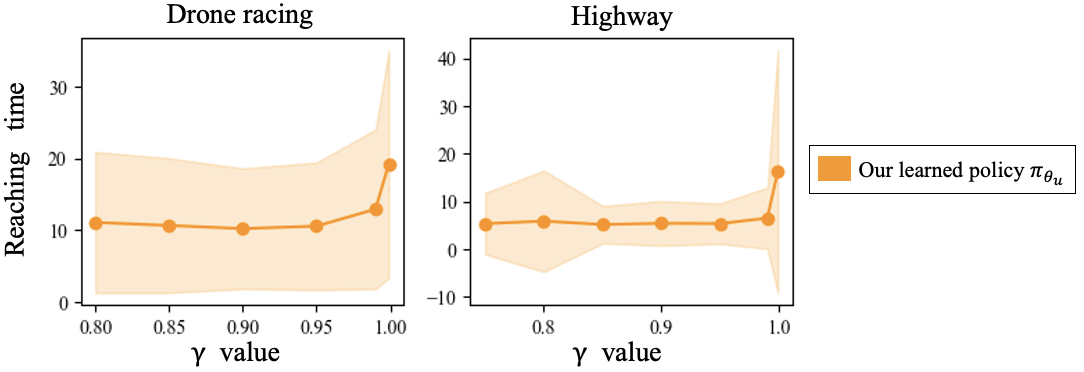}
    \caption{The average time taken for reaching the target set grows as $\gamma$ increasing. }\vspace{-1.5em}
    \label{fig:gamma effect}
\end{figure}

\subsection{Hypothesis 3: The Lipschitz continuity of our new value function appears to accelerate learning%The Lipschitz continuity of our new value function improves learning performance%The learning process of our method is more stable than the state-of-the-art reachability learning methods
}\vspace{-0.2em}
In Figure \ref{fig:learning 1}, %and \ref{fig:learning 2}, 
we compare the critic loss (defined in \eqref{eq:critic loss} to measure the Bellman equation error) under different Bellman equations and time-discount factor $\gamma$ values. %, using the same training hyper parameters. 
Our critic loss converges rapidly when $\gamma$ is chosen to ensure the Lipschitz continuity of our new value function. \vspace{-1em}%faster than the existing reachability learning method based on the classical RA value function in \cite{nguyen2024gameplay, li2023learning}. \vspace{-1em}%Empirically, when $\valuefunc$ is Lipschitz continuous, our critic loss converges at a sublinear rate, which matches the convergence analysis of DPG in \cite{xiong2022deterministic}. \vspace{-1em}
%During learning, the discontinuity of the baseline non-time-discounted RA value function is hard to fit, and it partially causes learning instability. 

\subsection{The trade-off of selecting a time-discount factor \texorpdfstring{$\gamma$}{}}\label{sec:gamma experiment}\vspace{-0.2em}
We summarize our result in Figures \ref{fig:gamma effect 1} and \ref{fig:gamma effect}. With a small $\gamma$, \revise{the learned RA sets can be conservative due to numerical errors, as an initial state whose optimal trajectory reaches the target set at a later stage may have near-zero time-discounted RA measures.} However, the optimal policy tends to drive the state to the target set rapidly, as depicted in Figure~\ref{fig:gamma effect}. \revise{Conversely, a large $\gamma$ can induce discontinuities in the value function, destabilizing learning and yielding suboptimal solutions.} In the drone racing and highway experiments, we find that $\gamma = 0.95$ ensures the Lipschitz continuity of $\valuefunc$, thereby enhancing learning efficiency, and also mitigates unnecessary conservatism. \vspace{-1em}
% Moreover, as shown in Figure~\ref{fig:gamma effect}, as $\gamma$ decreases, the average target set reaching time also decreasing, which empirically supports Theorem~\ref{thm:reaching time}, a small $\gamma$ makes 

%A small $\gamma$ introduces too much pessimism into the future value, and this affects learning quality. However, when $\gamma=1$, Bellman operator is not a contraction mapping and therefore the solution to Bellman equation is not unique. 
% \begin{tabular}{ |p{1.5cm}|p{1.6cm}|p{1.7cm}|p{1.9cm}| }
% \hline
% \multicolumn{4}{|c|}{Drone racing, RA set volume under different $\gamma$} \\
% \hline
% values& RA set vol. & QP cal. & Lipschitz cal.\\
% \hline
% $\gamma=1$ & AF &AFG & \\
% $\gamma=0.99$ & AX   & ALA &\\
% $\gamma=0.95$ &AL & ALB & \\
% $\gamma=0.9$    &DZ & DZA & \\
% \hline
% \end{tabular}

% \begin{tabular}{ |p{1.3cm}|p{1.6cm}|p{1.9cm}|p{1.9cm}| }
% \hline
% \multicolumn{4}{|c|}{Highway take-over, RA set volume under different $\gamma$} \\
% \hline
% values& RA set vol. & reaching time & FPR\\
% \hline
% $\gamma=1$ & AF &AFG & \\
% $\gamma=0.99$ & AX   & ALA &\\
% $\gamma=0.95$ &AL & ALB & \\
% $\gamma=0.9$    &DZ & DZA & \\
% \hline
% \end{tabular}

\section{\revise{Conclusion and Future Work}}%\vspace{-0.4em}
% We propose a new framework for learning trustworthy RA sets. Our method features a newly designed time-discounted RA value function which is provably Lipschitz continuous and has a Bellman operator that is a contraction mapping. This enhances computational efficiency by eliminating the need for the time-discount factor annealing process used in previous works. We use max-min DDPG to learn our value functions and propose two efficient methods to certify if a set of states can safely reach the target set despite any disturbance, with deterministic guarantees. We validate our methods in drone racing hardware experiments and highway take-over simulations. In future research, we plan to extend this work to online reachability learning and propose more efficient and less conservative reachability set certification methods.
\revise{We propose a new framework for learning trustworthy reach-avoid (RA) sets. Our method features a newly designed RA value function that offers improved computational efficiency.} %due to its Lipschitz continuity and contractive Bellman equation. 
We employ max-min DDPG to learn our value functions and propose two efficient methods to certify whether a set of states can safely reach the target set with deterministic guarantees. We validate our methods through drone racing hardware experiments and highway take-over simulations. \revise{Our certification methods can be performed in real time, but they rely on offline value function learning beforehand. Future research may explore online reachability learning, as well as more efficient RA set certification methods.%\revise{Our main limitation is offline value function learning. Future research could focus on online reachability learning and more efficient RA set certification methods with deterministic assurances.
}%\vspace{-1em}

\ifCLASSOPTIONcaptionsoff
  \newpage
\fi

\bibliographystyle{ieeetr}
\bibliography{references}

% trigger a \newpage just before the given reference
% number - used to balance the columns on the last page
% adjust value as needed - may need to be readjusted if
% the document is modified later
%\IEEEtriggeratref{8}
% The "triggered" command can be changed if desired:
%\IEEEtriggercmd{\enlargethispage{-5in}}

% references section

% can use a bibliography generated by BibTeX as a .bbl file
% BibTeX documentation can be easily obtained at:
% http://mirror.ctan.org/biblio/bibtex/contrib/doc/
% The IEEEtran BibTeX style support page is at:
% http://www.michaelshell.org/tex/ieeetran/bibtex/
%\bibliographystyle{IEEEtran}
% argument is your BibTeX string definitions and bibliography database(s)
%\bibliography{IEEEabrv,../bib/paper}
%
% <OR> manually copy in the resultant .bbl file
% set second argument of \begin to the number of references
% (used to reserve space for the reference number labels box)

% You can push biographies down or up by placing
% a \vfill before or after them. The appropriate
% use of \vfill depends on what kind of text is
% on the last page and whether or not the columns
% are being equalized.

%\vfill

% Can be used to pull up biographies so that the bottom of the last one
% is flush with the other column.
%\enlargethispage{-5in}

% that's all folks
\end{document}